\documentclass[runningheads]{llncs}
\usepackage[dvipsnames,x11names,svgnames]{xcolor}

\usepackage{hyperref}
\usepackage{amsmath,amssymb, amsfonts}

\usepackage{newtxtext,newtxmath}
\usepackage[disable]{todonotes}
\usepackage{graphicx}
\usepackage{algorithm2e}  
\usepackage{multirow}
\usepackage{multicol}
\usepackage{upgreek}
\usepackage{xfrac}
\usepackage{nicefrac}
\usepackage{csquotes}
\usepackage{adjustbox}

\usepackage{microtype}

\usepackage{listings}

\usepackage[caption=false]{subfig}
\usepackage{tikz}
\usetikzlibrary{patterns,arrows,backgrounds,calc,shapes,shadows,decorations.pathmorphing,decorations.pathreplacing,automata,shapes.multipart,positioning,shapes.geometric,fit,circuits,trees,shapes.gates.logic.US,fit,decorations.markings}

\newcommand{\mc}{\mathfrak{C}}
\newcommand{\transmatrix}{P}

\newcommand{\sinit}{s_I}
\newcommand{\bmsymbol}{\Phi}
\newcommand{\bmsymbola}[1]{\Phi_{#1}}

\newcommand{\mdp}{\mathfrak{M}}
\newcommand{\smdp}[1]{\mathfrak{M}_{#1}}
\newcommand{\act}{\textnormal{Act}}
\newcommand{\mdps}{\left( S, \,  \sinit , \, \act, \, \transmatrix \right)}
\newcommand{\mdpsprime}{\left( S', \,  \sinit , \, \act, \, \transmatrix' \right)}
\newcommand{\acts}[1]{\act \left( #1 \right)}

\newcommand{\probm}{\textnormal{Pr}}

\newcommand{\prevtbmax}[4]{\probm^{\textnormal{max}} \left( #2 \models \lozenge^{\leq #3} #4 \right) }
\newcommand{\prevtmax}[3]{\probm^{\textnormal{max}} \left( #2 \models \lozenge #3 \right) }

\newcommand{\bm}[1]{\bmsymbol \left( #1 \right)}
\newcommand{\bma}[2]{\bmsymbol_{#1} \left( #2 \right)}
\newcommand{\bmn}[2]{\bmsymbol^{#1} \left( #2 \right)}

\newcommand{\succs}{\textnormal{Succs}}

\newcommand{\reach}[1]{\textnormal{Reach}^{{}\leq #1}}

\newcommand{\zeroexp}{\mathbf{0}}
\newcommand{\oneexp}{\mathbf{1}}
\newcommand{\substright}{\right\rangle}
\newcommand{\substleft}{\left\langle}
\newcommand{\substarrow}{\mapsto}
\newcommand{\subst}[2]{\substleft #1 \substarrow #2 \substright}

\newcommand{\bmseq}{\bm{\zeroexp},\, \bmn{2}{\zeroexp},\, \bmn{3}{\zeroexp},\, \ldots}
\newcommand{\bmseqk}{\bm{\zeroexp},\, \bmn{2}{\zeroexp},\, \ldots,\, \bmn{k+1}{\zeroexp}}
\newcommand{\frameseq}{F_0,\, F_1,\, F_2,\, \ldots}
\newcommand{\frameseqk}{F_0,\, \ldots,\, F_k}
\newcommand{\frameseqkmo}{F_0,\, \ldots,\, F_{k-1}}
\newcommand{\frameseqkpo}{F_0,\, \ldots,\, F_{k+1}}

\SetKwRepeat{Do}{do}{while}
\newcommand{\changed}[1]{\textcolor{Crimson}{#1}}
\newcommand{\old}[1]{\textcolor{gray}{#1}}

\newcommand{\interface}[1]{\textcolor{DodgerBlue}{#1}}
\newcommand{\changedColor}{red\xspace}

\newcommand{\thresh}{\lambda}
\newcommand{\touched}{\textit{touched}}

\newcommand{\success}{\textit{success}}

\newcommand{\safeReturn}{\textit{safe}}
\newcommand{\subsystem}{\textit{subsystem}}
\newcommand{\oldSubsystem}{\textit{oldSubsystem}}

\newcommand{\algoname}[1]{\textup{\textsf{#1}}\xspace}
\newcommand{\boldalgoname}[1]{\textnormal{\textbf{\textsf{#1}}}\xspace}
\newcommand{\getprobs}{\heuristic}
\newcommand{\initOracle}{\algoname{Initialize}}
\newcommand{\enlarge}{\algoname{Enlarge}}
\newcommand{\createHeuristic}{\algoname{CreateHeuristic}}
\newcommand{\refineoracle}{\algoname{Refine}}
\newcommand{\checkrefutation}{\algoname{CheckRefutation}}
\newcommand{\icthree}{\algoname{IC3}}

\newcommand{\pricthree}{\algoname{PrIC3}}
\newcommand{\mainloop}{\ensuremath{\pricthree_{\heuristic}}\xspace}
\newcommand{\boldmainloop}{\ensuremath{\boldalgoname{PrIC3}_{\boldsymbol\heuristic}}\xspace}
\newcommand{\outerloop}{\ensuremath{\algoname{\pricthree}}\xspace}

\newcommand{\boldpricthree}{\boldalgoname{PrIC3}}

\newcommand{\blank}{\underline{\hspace{0.4cm}}}

\newcommand{\qtouched}{Q.\textnormal{touched}()}

\newcommand{\bad}{B}

\newcommand{\vars}{\textnormal{vars}}
\newcommand{\oracle}{\Omega}
\newcommand{\heuristic}{\ensuremath{\mathcal{H}}\xspace}

\newcommand{\prict}{\pricthree}
\newcommand{\prictcall}[4]{\pricthree_{#4}\left( #1, \, #2, \, #3 \right)}

\newcommand{\strengthenH}[1]{\ensuremath{\algoname{Strengthen}_{#1}}}
\newcommand{\strengthen}{\strengthenH{\heuristic}}
\newcommand{\strengthencall}[2]{\strengthenH{#2}\left( #1 \right)}

\newcommand{\propagate}{\algoname{Propagate}}
\newcommand{\propagatecall}[1]{\algoname{Propagate} \left( #1 \right)}

\newcommand{\qpop}{\textnormal{popMin}()}
\newcommand{\qpush}[1]{\textnormal{push}\left( #1 \right)}

\newcommand{\htriple}[3]{ \big\{\, #1 \,\big\} \, #2 \, \big\{\, #3 \,\big\} }

\newcommand{\icinv}[1]{\textsf{\upshape{PrIC3Inv}} \left( #1 \right)}

\newcommand{\genset}{G}
\newcommand{\genfunc}{p_\genset}

\newcommand{\fun}[3]{\ensuremath{#1 \colon #2 \to #3}}

\newcommand{\true}{\textup{\texttt{true}}}
\newcommand{\false}{\textup{\texttt{false}}}

\newcommand{\abs}[1]{\left| #1 \right|}
\newcommand{\eeq}{~{}={}~}

\newcommand{\lleq}{~{}\leq{}~}

\newcommand{\lleftarrow}{~{}\leftarrow{}~}

\newcommand{\RRightarrow}{~{}\Rightarrow{}~}
\newcommand{\iimplies}{~{}\implies{}~}

\newcommand{\qqiff}{\qquad{}\textnormal{iff}{}\qquad}

\newcommand{\qqand}{\qquad{}\textnormal{and}{}\qquad}
\newcommand{\qimplies}{\quad{}\textnormal{implies}{}\quad}
\newcommand{\qqimplies}{\qquad{}\textnormal{implies}{}\qquad}
\newcommand{\eewedge}{~{}\wedge{}~}

\newcommand{\eeleq}{~{}\leq{}~}

\newcommand{\Reals}{\mathbb{R}}
\newcommand{\Ints}{\mathbb{Z}}
\newcommand{\Bools}{\mathbb{B}}

\newcommand{\minval}[1]{\textsf{\upshape{min}} \left( #1 \right)}

\newcommand{\lfp}{\textnormal{lfp}}

\newcommand{\sebastianmargin}[1]{\todo[color=violet!35,size=\scriptsize,fancyline,author=Sebastian]{#1}\xspace}

\newcommand{\epmc}{\textnormal{EPMC}}
\newcommand{\prism}{\textnormal{Prism}}
\newcommand{\storm}{\textnormal{Storm}}
\newcommand{\modest}{\textnormal{modest}}

\newcommand{\eexp}[1]{\mathbf{E}^{#1}}

\definecolor{prismgreen}{HTML}{009900}
\definecolor{prismred}{HTML}{cc0000}
\definecolor{prismblue}{HTML}{0000cc}

\lstdefinelanguage{Prism}{
        basicstyle=\color{prismred}\scriptsize\ttfamily,
        literate=*	{0}{{\textcolor{prismblue}{0}}}{1}
			{1}{{\textcolor{prismblue}{1}}}{1}
			{2}{{\textcolor{prismblue}{2}}}{1}
			{3}{{\textcolor{prismblue}{3}}}{1}
			{4}{{\textcolor{prismblue}{4}}}{1}
			{5}{{\textcolor{prismblue}{5}}}{1}
			{6}{{\textcolor{prismblue}{6}}}{1}
			{7}{{\textcolor{prismblue}{7}}}{1}
			{8}{{\textcolor{prismblue}{8}}}{1}
			{9}{{\textcolor{prismblue}{9}}}{1}
			{.0}{{\textcolor{prismblue}{.0}}}{2}
			{.1}{{\textcolor{prismblue}{.1}}}{2}
			{.2}{{\textcolor{prismblue}{.2}}}{2}
			{.3}{{\textcolor{prismblue}{.3}}}{2}
			{.4}{{\textcolor{prismblue}{.4}}}{2}
			{.5}{{\textcolor{prismblue}{.5}}}{2}
			{.6}{{\textcolor{prismblue}{.6}}}{2}
			{.7}{{\textcolor{prismblue}{.7}}}{2}
			{.8}{{\textcolor{prismblue}{.8}}}{2}
			{.9}{{\textcolor{prismblue}{.9}}}{2}
			{[}{{\textcolor{black}{[}}}{1}
			{]}{{\textcolor{black}{]}}}{1}
			{+}{{\textcolor{black}{+}}}{1}
			{-}{{\textcolor{black}{-}}}{1}
			{=}{{\textcolor{black}{=}}}{1}
			{<}{{\textcolor{black}{<}}}{1}
			{>}{{\textcolor{black}{>}}}{1}
			{\&}{{\textcolor{black}{\&}}}{1}
			{|}{{\textcolor{black}{|}}}{1}
			{:}{{\textcolor{black}{:}}}{1}
			{;}{{\textcolor{black}{;}}}{1}
			{(}{{\textcolor{black}{(}}}{1}
			{)}{{\textcolor{black}{)}}}{1}
			{..}{{\textcolor{black}{..}}}{2},
        keywords= {bool,ceil,const,ctmc,double,dtmc,endinit,endmodule,endrewards, endsystem,F,false,floor,formula,G,global,I,init,int,label,max,mdp,min,module,nondeterministic,P,Pmin,Pmax,prob,probabilistic,rate,rewards,Rmin,Rmax,S,stochastic,system,true,U, option, either, assignment, relation, operation, hole, variable},
        keywordstyle={\bfseries\color{black}},
        numberstyle=\footnotesize\color{black},
        comment=[l] {//}, morecomment=[s]{/*}{*/},
        commentstyle= \color{prismgreen},
        tabsize=4,
        captionpos=b,
        escapechar=^,
        moredelim=[is][\color{orange}]{@}{@},
}

\renewcommand{\subsubsection}[1]{\medskip\noindent\textbf{#1}}

\renewcommand{\paragraph}[1]{\smallskip\noindent\emph{#1}}

\spnewtheorem{recovery}{Recovery Statement}{\bfseries}{\itshape}
\spnewtheorem{challenge}{Challenge}{\itshape}{\textnormal}
\LinesNumbered

\begin{document}
%
\title{\pricthree: Property Directed Reachability for MDPs\thanks{This work has been supported by the ERC Advanced Grant 787914 (FRAPPANT), NSF grants 1545126 (VeHICaL) and 1646208, the DARPA Assured Autonomy program,  Berkeley Deep Drive, and by Toyota under the iCyPhy center.}}
%
\titlerunning{\pricthree: Property Directed Reachability for MDPs}
%
\author{Kevin Batz\inst{1} \and
Sebastian Junges\inst{2}\and
Benjamin Lucien Kaminski\inst{3}\and\\
Joost-Pieter Katoen\inst{1}
\and
Christoph Matheja\inst{4}\and
Philipp Schr\"{o}er\inst{1}
}

\authorrunning{K. Batz, S. Junges, B.~Kaminski, J.-P.~Katoen, C.~Matheja, P.~Schr\"{o}er}

%
\institute{RWTH Aachen University, Aachen, Germany
 \and
 University of California, Berkeley, USA
 \and
 University College London, United Kingdom
 \and
 ETH Z\"{u}rich, Z\"{u}rich, Switzerland
 }

%
\maketitle              
\begin{abstract}
IC3 has been a leap forward in symbolic model checking. This paper proposes PrIC3 (pronounced pricy-three), a conservative extension of IC3 to symbolic model checking of MDPs. 
Our main focus is to develop the theory underlying PrIC3.
Alongside, we present a first implementation of PrIC3 including the key ingredients from IC3 such as generalization, repushing, and propagation. 
\end{abstract}
\section{Introduction}
\label{sec:intro}

\paragraph{\icthree.}
Also known as property-directed reachability (PDR)~\cite{DBLP:conf/fmcad/EenMB11}, \icthree~\cite{DBLP:conf/vmcai/Bradley11} is a symbolic approach for verifying finite transition systems (TSs) against safety properties like \enquote{\emph{bad~states are unreachable}}.
It combines bounded model checking (BMC)~\cite{DBLP:series/faia/Biere09} and inductive invariant generation.
Put shortly, \icthree either proves that a set $B$ of bad states is \emph{un}reachable by finding a set of non-$B$ states closed under reachability---called an \emph{inductive invariant}---or refutes reachability of $B$ by a \emph{counterexample} path reaching $B$.
Rather than unrolling the transition relation (as in BMC), \icthree attempts to incrementally strengthen the invariant \enquote{no state in $B$ is reachable} into an inductive one. 
In addition, it applies aggressive abstraction to the explored state space, so-called generalization~\cite{DBLP:conf/fmcad/HassanBS13}.
These aspects together with the enormous advances in modern SAT solvers have led to \icthree's success.
\icthree has been extended~\cite{DBLP:conf/sat/HoderB12,DBLP:conf/fmcad/GurfinkelI15} and adapted to software verification~\cite{DBLP:journals/fmsd/CimattiGMT16,sttt_siemens2020}.
This paper develops a \emph{quantitative} \icthree framework for probabilistic models.

\paragraph{MDPs.}
Markov decision processes (MDPs) extend TSs with discrete probabilistic choices.
They are central in planning, AI as well as in modeling randomized distributed algorithms.
A key question in verifying MDPs is \emph{quantitative} reachability: \enquote{\emph{is~the (maximal) probability to reach $B$ at most $\lambda$?}}. 
Quantitative reachability~\cite{DBLP:reference/mc/BaierAFK18,DBLP:series/lncs/BaierHK19} reduces to solving linear programs (LPs).
Various tools support MDP model checking, e.g., \prism~\cite{DBLP:conf/cav/KwiatkowskaNP11}, \storm~\cite{DBLP:conf/cav/DehnertJK017}, \modest~\cite{DBLP:conf/tacas/HartmannsH14}, and \epmc~\cite{DBLP:conf/fm/HahnLSTZ14}. 
The LPs are mostly solved using (variants of) value iteration~\cite{DBLP:journals/tcs/HaddadM18,DBLP:conf/cav/Baier0L0W17,DBLP:conf/cav/QuatmannK18,DBLP:journals/corr/abs-1910-01100}. 
Symbolic BDD-based MDP model checking originated two decades ago~\cite{DBLP:conf/tacas/AlfaroKNPS00} and is rather successful.

\paragraph{Towards \icthree for MDPs.}
Despite the success of BDD-based symbolic methods in tools like \prism, \icthree has not penetrated probabilistic model checking yet.
The success of \icthree and the importance of quantitative reachability in probabilistic model checking raises the question \emph{whether and how \icthree can be adapted---not just utilized---to reason about quantitative reachability in MDPs}.
This paper addresses the challenges
of answering this question. It extends \icthree in several dimensions to overcome these hurdles, making \pricthree---to our knowledge---\emph{the first \icthree framework for quantitative reachability in MDPs}\footnote{Recently, (standard) \icthree for TSs was \emph{utilized} in model checking Markov chains~\cite{DBLP:journals/corr/abs-1909-08017} to on-the-fly compute the states that cannot reach $B$.}.
Notably, \pricthree is conservative: For a threshold $\lambda = 0$, \pricthree{} solves the same qualitative problem \emph{and behaves (almost) the same as standard \icthree{}}. 
Our main contribution is developing the theory underlying \pricthree, which is accompanied by a proof-of-concept implementation.
\begin{challenge}[Leaving the Boolean domain]
\label{challenge:domain}
\icthree iteratively computes \emph{frames}, which are over-approximations of sets of states that can reach $B$ in a bounded number of steps.
For MDPs, Boolean reachability becomes a \emph{quantitative reachability probability}.
This requires a shift: frames become real-valued functions rather than sets of states.
Thus, there are infinitely many possible frames---even for finite-state MDPs---just as for infinite-state software~\cite{DBLP:journals/fmsd/CimattiGMT16,sttt_siemens2020} and hybrid systems~\cite{DBLP:conf/vmcai/SuenagaI20}. 
Additionally, whereas in TSs a state reachable within $k$ steps remains reachable on increasing $k$, the reachability probability in MDPs may increase.
This complicates ensuring termination of an \icthree algorithm for MDPs.\hfill$\triangle$
\end{challenge}%
\vspace*{-1.1\baselineskip}
\begin{challenge}[Counterexamples $\neq$ single paths]
\label{challenge:counterexamples}
For TSs, a single cycle-free path\footnote{In~\cite{DBLP:conf/sat/HoderB12}, tree-like counterexamples are used for non-linear predicate transformers in \icthree.} to $B$ suffices to refute that \enquote{\emph{$B$ is not reachable}}.
This is not true in the probabilistic setting~\cite{DBLP:journals/tse/HanKD09}.
Instead, proving that the probability of reaching $B$ exceeds the threshold $\lambda$ requires \emph{a set of possibly  cyclic paths}---e.g., represented as a sub-MDP~\cite{DBLP:journals/tocl/ChadhaV10}---whose probability mass exceeds $\lambda$.
Handling sets of paths as counterexamples in the context of \icthree is new.\hfill$\triangle$
\end{challenge}%
\vspace*{-1.1\baselineskip}
\begin{challenge}[Strengthening]
\label{challenge:strengthen}
This key \icthree technique intuitively turns a proof obligation of type (\textsc{i}) \enquote{state $s$ is unreachable from the initial state $\sinit$} into type~(\textsc{ii})~\enquote{$s$'s~\mbox{\emph{predecessors}} are unreachable from $\sinit$}.
A first issue is that in the quantitative setting, the standard characterization of reachability probabilities in MDPs (the Bellman equations) inherently \emph{reverses} the direction of reasoning (cf.\ \enquote{reverse} \icthree~\cite{DBLP:conf/mbmv/SeufertS17}): 
Hence, strengthening turns (\textsc{i}) \enquote{$s$ cannot reach $\bad$} into (\textsc{ii}) \enquote{$s$'s \emph{successors} cannot reach $\bad$}.%

A much more challenging issue, however, is that in the quantitative setting obligations of type (\textsc{i})  read \enquote{$s$ is reachable \emph{with at most probability $\delta$}}.
However, the strengthened type~(\textsc{ii}) obligation must then read:
\enquote{\emph{the weighted sum over the reachability probabilities of the successors of $s$} is at most $\delta$}.
In general, 
there are infinitely many possible choices of subobligations for the successors of $s$ in order to satisfy the original obligation, because---grossly simplified---there are infinitely many possibilities for $a$ and $b$ to satisfy weighted sums such as $\tfrac{1}{3} a + \tfrac{2}{3} b \leq \delta$.
While we only need one choice of subobligations, picking a \emph{good} one is approximately as hard as solving the entire problem altogether.
We hence require a heuristic, which is guided by a \emph{user-provided oracle}.\hfill$\triangle$
\end{challenge}%
\vspace*{-1.1\baselineskip}
\begin{challenge}[Generalization] 
\label{challenge:generalization}
\enquote{One of the key components of \icthree is [inductive] generalization}~\cite{DBLP:conf/vmcai/Bradley11}.
Generalization~\cite{DBLP:conf/fmcad/HassanBS13} abstracts single states. It makes \icthree scale, but is \emph{not} essential for correctness.
To facilitate generalization, systems should be encoded symbolically, i.e., integer-valued program variables describe states.
Frames thus map variables to probabilities. 
A first aspect is how to effectively present them to an SMT-solver.
Conceptually,  we use uninterpreted functions and universal quantifiers (encoding program behavior) together with linear real arithmetic to encode the weighted sums occurring when reasoning about probabilities.
A second aspect is more fundamental:
Abstractly, \icthree's generalization guesses an unreachable set of states.
We, however, need to guess this set \emph{and} a probability for each state. 
To be effective, these guesses should moreover eventually yield an inductive frame, which is often highly nonlinear.
We propose three SMT-guided interpolation variants for guessing these maps.\hfill$\triangle$
\end{challenge}%
\vspace*{-.5\baselineskip}
\paragraph{Structure of this paper.}
We develop \pricthree gradually:
We explain the underlying rationale in Sect.~\ref{sec:setup}. 
We also describe the core of \pricthree---called \mainloop---which resembles closely the main loop of standard \icthree, but uses adapted frames and termination criteria (Chall.~\ref{challenge:domain}). 
In line with Chall.~\ref{challenge:strengthen}, \mainloop is parameterized by a heuristic $\heuristic$ which is applied whenever we need to select one out of infinitely many probabilities. 
No requirements on the quality of \heuristic are imposed.
$\mainloop$ is \emph{sound} and always terminates: If it returns $\true$, then the maximal reachability probability is bounded by $\lambda$. 
Without additional assumptions about \heuristic, \mainloop is \emph{incomplete}: on returning $\false$, it is unknown whether the returned subMDP is indeed a counterexample (Chall.~\ref{challenge:counterexamples}).
Sect.~\ref{sec:strengthen} details strengthening (Chall.~\ref{challenge:strengthen}).
Sect.~\ref{sec:outermost_loop} presents a sound \emph{and} complete algorithm \outerloop on top of \mainloop.
Sect.~\ref{sec:practical} presents a prototype, discusses our chosen heuristics, and addresses Chall.~\ref{challenge:generalization}.
Sect.~\ref{sec:experiments} shows some encouraging experiments, but also illustrates need for further progress.%

\subsubsection{Related Work.}
%
%
Just like \icthree has been a symbiosis of different approaches, \pricthree has been inspired by several existing techniques from the verification of probabilistic systems.

\paragraph{BMC.}
Adaptions of BMC to Markov chains (MCs) with a dedicated treatment of cycles 
have been pursued in~\cite{DBLP:conf/vmcai/WimmerBB09}. 
The encoding in~\cite{DBLP:conf/hybrid/FranzleHT08} annotates sub-formulae with probabilities.
The integrated SAT solving process implicitly unrolls all paths leading to an exponential blow-up.
In \cite{DBLP:conf/qest/RabeWKYH14}, this is circumvented by grouping paths, discretizing them, and using an encoding with quantifiers and bit-vectors, but without numerical values.
Recently, \cite{DBLP:journals/corr/abs-1903-09354} extends this idea to a PAC algorithm by purely propositional encodings and (approximate) model counting \cite{DBLP:conf/ijcai/ChakrabortyFMV15}. 
These approaches focus on MCs and are not mature yet.

\paragraph{Invariant synthesis.}
Quantitative loop invariants are key in analyzing \emph{probabilistic programs} whose operational semantics are (possibly infinite) MDPs~\cite{DBLP:journals/pe/GretzKM14}.
A quantitative invariant $I$ maps states to probabilities.
$I$ is shown to be an invariant by comparing $I$ to the result of applying the MDP's Bellman operator to~$I$.
Existing approaches for invariant synthesis are, e.g., based on weakest pre-expectations~\cite{DBLP:conf/stoc/Kozen83,DBLP:series/mcs/McIverM05,DBLP:journals/jacm/KaminskiKMO18,thesis:kaminski,DBLP:journals/pacmpl/HarkKGK20}, template-based constraint solving~\cite{DBLP:conf/qest/GretzKM13}, 
 notions of martingales~\cite{DBLP:conf/cav/ChakarovS13,DBLP:conf/cav/BartheEFH16,DBLP:journals/pacmpl/AgrawalC018,DBLP:conf/atva/TakisakaOUH18},
and solving recurrence relations~\cite{DBLP:conf/atva/BartocciKS19}.
All but the last technique require user guidance.


\paragraph{Abstraction.}
To combat state-space explosion, abstraction is often employed.
CEGAR for MDPs~\cite{DBLP:conf/cav/HermannsWZ08} deals with explicit sets of paths as counterexamples. 
Game-based abstraction~\cite{DBLP:journals/fmsd/KattenbeltKNP10,DBLP:conf/tacas/HahnHWZ10} and partial exploration~\cite{DBLP:conf/atva/BrazdilCCFKKPU14} exploit that not all paths have to be explored to prove bounds on reachability probabilities. 

\paragraph{Statistical methods and (deep) reinforcement learning.}
Finally, an avenue that avoids storing a (complete) model are simulation-based approaches (statistical model checking~\cite{DBLP:journals/tomacs/AghaP18}) and variants of reinforcement learning, possibly with  neural networks. 
For MDPs, these approaches yield weak statistical guarantees~\cite{DBLP:conf/isola/DArgenioHS18}, but may provide good oracles.

\section{Problem Statement}
\label{sec:preliminaries}
%

Our aim is to prove that the \emph{maximal probability} of \emph{reaching} a \emph{set $\bad$ of bad states} from the initial state $\sinit$ of a \emph{Markov decision process} $\mdp$ is at most some \emph{threshold}~$\thresh$.
Below, we give a formal description of our problem. We refer to~\cite{puterman1994markov,BK08} for a thorough introduction.%
\begin{definition}[MDPs]
\label{def:mdps}
	 A 
	  \emph{Markov decision process} (\emph{MDP}) is a tuple $\mdp = \mdps$,
	 where $S$ is a finite set of \emph{states}, $\sinit \in S$ is the \emph{initial state}, $\act$ is a finite
	 set of \emph{actions}, and $\fun{\transmatrix}{S \times \act \times S}{[0,1]}$ is a \emph{transition probability function}. 
	 For state $s$, let $\acts{s} = \left\{ a \in \act ~\mid~ \exists s' \in S \colon \transmatrix(s,a,s') > 0\right\}$ be the  \emph{enabled actions} at $s$. 
     For all states $s \in S$, 
     we require $|\acts{s}| \geq 1$ 
     and $\sum_{s'\in S} \transmatrix(s,a,s') = 1$.
	 \hfill $\triangle$
\end{definition}%
For this paper, we fix an MDP $\mdp = \mdps$, a set of \mbox{\emph{bad states} $\bad \subseteq S$}, and a threshold $\thresh \in [0, 1]$.
The \emph{maximal}\footnote{Maximal with respect to all possible resolutions of nondeterminism in the MDP.} \emph{(unbounded) reachability probability} to eventually reach a state in $\bad$ from a state $s$ is denoted by $\prevtmax{\mc}{s}{\bad}$.
We characterize
$\prevtmax{\mc}{s}{\bad}$ using the so-called \emph{Bellman operator}. 
Let $M^N$ denote the set of functions from $N$ to $M$.
Anticipating \icthree terminology, we call a function $F \in [0,1]^{S}$ a \emph{frame}. 
We denote by $F[s]$ the evaluation of frame $F$ for state $s$.%
\begin{definition}[Bellman Operator]
	 For a set of actions $A \subseteq \act$, we define the  \emph{Bellman operator for $A$} as  a frame transformer $\fun{\bmsymbola{A}}{[0,1]^{S}}{[0,1]^{S}}$ with%
	\begin{align*}
		\bma{A}{F}[s] \eeq \begin{cases}
			1, & ~\text{if}~ s \in \bad \\
			%
			%
			\max\limits_{a \in A}~  \sum\limits_{s' \in S} \transmatrix(s, a, s') \cdot F[s']~, & ~\text{if}~ s \notin \bad~.
     \end{cases}
   \end{align*}
   We write $\bmsymbola{a}$ for~$\bmsymbola{ \{ a \}}$, $\bmsymbol$ for $\bmsymbola{\act}$, and call $\bmsymbol$ simply \emph{the Bellman operator}. \hfill $\triangle$
\end{definition}%
For every state $s$, the maximal reachability probability $\prevtmax{\mc}{s}{\bad}$ is then given by the 
least fixed point of the Bellman operator $\Phi$. 
That is,
\begin{align*}
	\forall\, s\colon\quad \prevtmax{\mc}{s}{\bad} \eeq \bigl(\lfp~ \bmsymbol \bigr)[s]~,
\end{align*}%
where the underlying partial order on frames is a complete lattice with ordering
\begin{align*}
	F_1 \lleq F_2 \qqiff \forall\, s \in S\colon \quad F_1[s] \lleq F_2[s]~.
\end{align*}%
%
%
In terms of the Bellman operator, our formal problem statement reads as follows:
%
%
%
%
\begin{center}\fbox{%
	\parbox{.98\textwidth}{%
		\centering{
			Given an MDP $\mdp$ with initial state $\sinit$, 
			a set~$\bad$ of bad states, 
			and a \emph{threshold} $\thresh \in [0,1]$,%
			\begin{align*}
				\text{prove or refute that}\qquad 
				\prevtmax{\mc}{\sinit}{\bad}\eeq \bigl(\lfp~ \bmsymbol \bigr)[\sinit] 
				\quad\leq{}\quad \thresh~.
			\end{align*}%
			\vspace{-1.4em}
		}%
	}%
}%
\end{center}
Whenever $\prevtmax{\mc}{\sinit}{\bad} \leq \thresh$ indeed holds, we say that the MDP $\mdp$ is \emph{safe} (with respect to the set of bad states $\bad$ and threshold $\thresh$); otherwise, we call it \emph{unsafe}.%
\begin{recovery}
	\label{recovery:1}
For $\thresh = 0$, our problem statement is equivalent to the \emph{qualitative reachability} problem solved by (reverse) standard \icthree, i.e, prove or refute that all bad states in $\bad$ are \emph{unreachable} from the initial state $\sinit$.
\end{recovery}%
\begin{example}
\label{ex:mdp}
The MDP $\mdp$ in Fig.~\ref{fig:runningex} consists of 6 states with initial state $s_0$ and bad states $\bad = \{ s_5 \}$.
In $s_2$, actions $a$ and $b$ are enabled; in all other states, one unlabeled action is enabled.
We have $\prevtmax{\mc}{s_0}{\bad} = \nicefrac{2}{3}$.
Hence, $\mdp$ is safe for all thresholds $\lambda \geq \nicefrac{2}{3}$ and unsafe for $\lambda < \nicefrac{2}{3}$.
In particular, $\mdp$ is unsafe for $\lambda = 0$ as $s_5$ is \emph{reachable} from $s_0$.\hfill$\triangle$
\end{example}%
\begin{figure}[t]
\centering
%
%
\begin{tikzpicture}[every node/.style={font={\scriptsize}}, shorten >=1pt,nodestyle/.style={draw,circle,font={\scriptsize}},actstyle/.style={draw,circle,fill=black,inner sep=1pt}]
	\draw[help lines,use as bounding box,white] (-5.1,-.5) grid (6.7, 0.7);
	\node[nodestyle] (s0) {$s_0$};
	\node[nodestyle,right=1cm of s0] (s1) {$s_1$};
	\node[nodestyle,left=1cm of s0] (s2) {$s_2$};
	\node[nodestyle,right=1cm of s1] (s3) {$s_3$};
	\node[nodestyle,left=2cm of s2] (s4) {$s_4$};
    \node[nodestyle,fill=Crimson,right=2cm of s3] (s5) {$\textcolor{white}{s_5}$};
	\node[actstyle,above=0.2cm of s2] (b) {};
	\node[actstyle,left=1cm of s2] (a) {};
	
	\draw[<-] (s0) edge ++(0,0.6);
	
	\draw[->] (s0) edge[bend left=20] node[above] {$\nicefrac{1}{2}$} (s1);
	\draw[->] (s0) edge node[above, pos=0.6] {$\nicefrac{1}{2}$} (s2);
	
	\draw[-] (s2) -- node[above] {$a$} (a); 
	\draw[-] (s2) --  node[left] {$b$} (b); 
	
	\draw[->] (b) edge[bend left=10] node[above,pos=0.6] {$1$} (s0);

	\draw[->] (s1) edge[bend left=20] node[above] {$\nicefrac{1}{2}$} (s0);
	\draw[->] (s1) edge node[above] {$\nicefrac{1}{2}$} (s3);
	
	\draw[->] (s3) -- node[right,pos=0.5] {$\nicefrac{1}{3}$} +(0,0.7) -| (s4);
	\draw[->] (s3) edge node[above] {$\nicefrac{2}{3}$} (s5);
	\draw[->] (a) edge node[above] {$\nicefrac{1}{2}$} (s4);
	\draw[->] (a) -- node[left,pos=0.5] {$\nicefrac{1}{2}$} +(0,-0.5) -| (s5);

	\draw[->] (s4) edge[loop left] node[left] {$1$} (s4);
	\draw[->] (s5) edge[loop right] node[right] {$1$} (s5);
\end{tikzpicture}%
 \caption{The MDP $\mdp$ serving as a running example.}	
 \label{fig:runningex}
\end{figure}
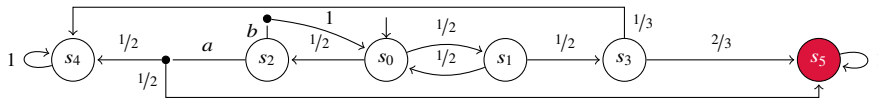%
\section{The Core \boldpricthree Algorithm}
\label{sec:setup}
%

The purpose of \pricthree is to prove or refute that the maximal probability to reach a bad state in $\bad$ from the initial state $\sinit$ of the MDP $\mdp$ is at most~$\thresh$.
In this section, we explain the rationale underlying \pricthree.
Moreover, we describe the core of \pricthree---called \mainloop---which bears close resemblance to the main loop of standard \icthree for TSs.

Because of the inherent direction of the Bellman operator, we build \pricthree on \emph{reverse} \icthree~\cite{DBLP:conf/mbmv/SeufertS17}, cf.~Chall.~\ref{challenge:strengthen}.
Reversing constitutes a shift from reasoning along the direction \emph{initial-to-bad} to \emph{bad-to-initial}.
While this shift is mostly \emph{inessential} to the fundamentals underlying \icthree, the reverse direction is unswayable in the probabilistic setting.
Whenever we draw a connection to standard \icthree, we thus generally mean \emph{reverse} \icthree.

\subsection{Inductive Frames}
\label{sec:pric3-basic-principle}

\icthree for TSs operates on
(\emph{quali}tative) frames representing sets of states of the TS at hand.
A frame $F$ can hence be thought of as a mapping\footnote{In \icthree, frames are typically characterized by logical formulae. To understand \icthree's fundamental principle, however, we prefer to think of frames as functions in $\{0,1\}^{S}$ partially ordered by $\leq$.} from states to~\mbox{$\{0, 1\}$}. 
In \pricthree for MDPs, we need to move from a Boolean to a quantitative regime. 
Hence, a (\emph{quanti}tative) frame is a mapping from states to probabilities in $[0,1]$.

For a given TS, consider the frame transformer $T$ that adds to a given input frame~$F'$ all bad states in $\bad$ and all predecessors of the states contained in $F'$.
The rationale of standard (reverse) \icthree is to find a frame $F \in \{0,1\}^{S}$ such that (\textsc{i})~the initial state $\sinit$ does not belong to $F$ and (\textsc{ii}) applying $T$ takes us down in the partial order on frames, i.e.,
\begin{align*}
   (\textsc{i})\quad F[\sinit] \eeq 0 \qqand (\textsc{ii})\quad T(F) \lleq F~.
\end{align*}%
Intuitively, (\textsc{i}) postulates the \emph{hypothesis} that $\sinit$ cannot reach $\bad$ and (\textsc{ii})~expresses that $F$ is closed under adding bad states and taking predecessors, thus affirming the hypothesis.

Analogously, the rationale of \pricthree is to find a frame $F \in [0,1]^{S}$ such that (\textsc{i})~$F$~postulates that the probability of $\sinit$ to reach $\bad$ is at most the threshold $\thresh$ and (\textsc{ii}) applying the Bellman operator $\bmsymbol$ to $F$ takes us down in the partial order on frames, i.e.,
\begin{align*}
	(\textsc{i})\quad F[\sinit] \lleq \thresh \qqand (\textsc{ii})\quad \bmsymbol(F) \lleq F~.
\end{align*}
Frames satisfying the above conditions are called \emph{inductive invariants} in~\icthree. 
We adopt this terminology.
By \emph{Park's Lemma}~\cite{park1969fixpoint}, which in our setting reads
\begin{align*}
	\bmsymbol(F) \lleq F \qimplies \lfp~\bmsymbol \lleq F~,
\end{align*}
an inductive invariant $F$ would indeed \emph{witness} that $\prevtmax{\mc}{\sinit}{\bad} \leq \thresh$, because
\begin{align*}
	\prevtmax{\mc}{\sinit}{\bad} \eeq \bigl( \lfp~ \bmsymbol \bigr)[\sinit] \lleq F[\sinit] \lleq \thresh~.
\end{align*}%
If no inductive invariant exists, then standard \icthree will find a counterexample: a \emph{path} from the initial state $\sinit$ to a bad state in $\bad$, which serves as a witness to refute.
Analogously, \pricthree will find a counterexample, but of a different kind:
Since single paths are insufficient as counterexamples in the probabilistic realm (Chall.~\ref{challenge:counterexamples}), \pricthree will instead find a \emph{subsystem} of states of the MDP witnessing $\prevtmax{\mc}{\sinit}{\bad} > \thresh$.

\subsection{The \boldpricthree Invariants}
\label{sec:pric3invariants}
Analogously to standard \icthree, \pricthree aims to find the inductive invariant by maintaining a \emph{sequence of frames} $F_0 \leq F_1 \leq F_2 \leq \ldots$ such that $F_i[s]$ overapproximates the maximal probability of reaching $B$ from $s$ within \emph{at most $i$ steps}.
This \emph{$i$-step-bounded reachability probability} $\prevtbmax{\mdp}{s}{i}{\bad}$
can be characterized using the Bellman operator:
$\bm{\zeroexp}$ is the $0$-step probability; it is $1$ for every $s \in \bad$ and $0$ otherwise.
For any $i \geq 0$, we have
\begin{align*}
    \prevtbmax{\mdp}{s}{i}{\bad} 
    \eeq \Bigl( \bmn{i}{\vphantom{\bigl(}\bm{\zeroexp}} \Bigr)[s]
    \eeq \Bigl( \bmn{i+1}{\zeroexp} \Bigr)[s]~,
\end{align*}%
where $\zeroexp$, the frame that maps every state to $0$, is the least frame of the underlying complete lattice.
For a finite MDP, the \emph{unbounded} reachability probability is then given by the limit
\begin{align*}
	\prevtmax{\mc}{s}{\bad} \eeq \bigl(\lfp~ \bmsymbol \bigr)[s]~ \stackrel{(*)}{=}{}~  \left( \hspace{.125ex} \lim_{n\rightarrow\infty} \bmn{n}{\zeroexp} \right)[s] 
    \eeq \lim_{n\rightarrow\infty} \prevtbmax{\mdp}{s}{n}{\bad}~,
\end{align*}%
where $(*)$ is a consequence of the well-known Kleene fixed point theorem~\cite{DBLP:journals/ipl/LassezNS82}.

The sequence $F_0 \leq F_1 \leq F_2 \leq \ldots$ 
maintained by \pricthree should frame-wise 
overapproximate the increasing sequence $\bm{\zeroexp} \leq \bmn{2}{\zeroexp} \leq \bmn{3}{\zeroexp} \ldots$.
Pictorially: 
\begin{center}%
	\renewcommand{\arraystretch}{1.2}%
    \begin{tabular}{c@{\quad}c@{\quad}c@{\qquad}c@{\qquad}c@{\qquad}c@{\qquad}c@{\qquad}c@{\qquad}c@{\qquad}c@{\qquad}c}
            & & $F_0$ & $\leq$ & $F_1$ & $\leq$ & $F_2$ & $\leq$ & $\ldots$  & $\leq$ & $F_k$ \\
            & & \rotatebox{90}{$\leq$} & ~ & \rotatebox{90}{$\leq$} & ~ & \rotatebox{90}{$\leq$} & ~ & ~ & ~ & \rotatebox{90}{$\leq$} \\
            $\zeroexp$ & $\leq$ & $\bm{\zeroexp}$ & $\leq$ & $\bmn{2}{\zeroexp}$ & $\leq$ & $\bmn{3}{\zeroexp}$ & $\leq$ & $\ldots$  & $\leq$ & $\bmn{k+1}{\zeroexp}$
	\end{tabular}%
	\renewcommand{\arraystretch}{1}%
\end{center}%
However, the sequence $\bmseq$ will never explicitly be known to \pricthree.
Instead, \pricthree will ensure the above frame-wise overapproximation property implicitly by enforcing the so-called \emph{\pricthree invariants} on the frame sequence $\frameseq$.
Apart from allowing for a threshold $0 \leq \thresh \leq 1$ on the maximal reachability probability,
these invariants coincide with the standard \icthree invariants (where $\lambda = 0$ is fixed).
Formally:%
\begin{definition}[\boldpricthree Invariants]%
\label{def:invariants}
	Frames $\frameseqk$, for $k \geq 0$, satisfy the \emph{\pricthree invariants}, a fact we will denote by $\icinv{\frameseqk}$, if all of the following hold:\\%

	\vspace{-.5em}\renewcommand{\arraystretch}{1.25}%
	\begin{tabular}{l@{\qquad}l}
        \textnormal{\textbf{1.}}\quad\emph{\textbf{Initiality:}} & $F_0 \eeq \bm{\zeroexp}$ \\
		\textnormal{\textbf{2.}}\quad\emph{\textbf{Chain Property:}} & $\forall\, 0 \leq i < k \colon \quad F_i \lleq F_{i+1}$ \\
		\textnormal{\textbf{3.}}\quad\emph{\textbf{Frame-safety:}} & $\forall\, 0 \leq i \leq k \colon\quad F_i[\sinit] \lleq \thresh$ \\
		\textnormal{\textbf{4.}}\quad\emph{\textbf{Relative Inductivity:}} & $\forall\, 0 \leq i < k \colon \quad \bm{F_i} \lleq F_{i+1}$
	\end{tabular}%
	\renewcommand{\arraystretch}{1}\\[-1em]%
	\noindent{}$~$\hfill $\triangle$
\end{definition}%
%
%
The \pricthree invariants enforce the above picture:
The \emph{chain property} ensures $F_0 \leq F_1 \leq \ldots \leq F_k$.
We have $\bm{\zeroexp} = F_0 \leq F_0$ by \emph{initiality}.
Assuming $\bmn{i+1}{\zeroexp} \leq F_i$ as induction hypothesis, monotonicity of $\bmsymbol$ and \emph{relative inductivity} imply $\bmn{i+2}{\zeroexp} \leq \bmsymbol(F_i) \leq F_{i+1}$.
%
%
%
%

By overapproximating $\bmseqk$, the frames $F_0,\,\ldots,\, F_k$ in effect bound the maximal step-bounded reachability probability of every state:%
\begin{lemma}
\label{lem:ic3_inv_overapprox}
	Let frames $\frameseqk$ satisfy the \pricthree invariants. 
	Then%
	\begin{align*}
		%
		%
        \forall\, s~~
        \forall\, i \leq k\colon \quad
        \prevtbmax{\mdp}{s}{i}{\bad} \lleq F_i[s].
	\end{align*}%
\end{lemma}%
%
%
%
In particular, Lem.~\ref{lem:ic3_inv_overapprox} together with \emph{frame-safety} ensures that the maximal step-bounded reachability probability of the \emph{initial state} $\sinit$ to reach $\bad$ is at most the threshold $\thresh$.

As for proving that the \emph{unbounded} reachability probability is also at most $\thresh$, it suffices to find two consecutive frames, say $F_i$ and $F_{i+1}$, that coincide:%
\begin{lemma}
\label{lem:ic3_inv_eqframes}
	Let frames $\frameseqk$ satisfy the \pricthree invariants. 
	Then%
	\begin{align*}
	    \exists\, i < k \colon \quad F_i \eeq F_{i+1}
		\qqimplies
		\prevtmax{\mdp}{\sinit}{\bad} \lleq \thresh~.
	\end{align*}%
\end{lemma}%
\begin{proof}
    $F_i = F_{i + 1}$ and \emph{relative inductivity} yield $\bmsymbol(F_i) \leq F_{i+1} = F_i$, rendering $F_i$ \emph{inductive}. By Park's lemma (cf.~Section~\ref{sec:pric3-basic-principle}), we obtain $\lfp~\bmsymbol \leq F_i$ and---by \emph{frame-safety}---conclude
	\begin{align*}
		\prevtmax{\mdp}{\sinit}{\bad} 
		\eeq 
		\bigl( \lfp~\bmsymbol \bigr) [\sinit] 
		\lleq F_i [\sinit]  
		\lleq
		\thresh~.
        \tag*{\qed}
	\end{align*}
	%
\end{proof}%

\subsection{Operationalizing the \boldpricthree Invariants for Proving Safety}
\label{sec:pric3-main-loop}
\noindent
Lem.~\ref{lem:ic3_inv_eqframes} gives us a clear angle of attack for \emph{proving} an MDP safe: 
Repeatedly add and refine frames approximating step-bounded reachability probabilities for more and more steps while enforcing the \pricthree invariants (cf.\ Def.~\ref{sec:pric3invariants}) until two \mbox{consecutive frames coincide}.

%
%
\begin{algorithm}[t]
\KwData{\changed{MDP $\mdp$},\quad set of bad states $\bad$,\quad \changed{threshold $\lambda$}}
\KwResult{$\true$ or $\false$ and a subset of the states of $\mdp$}
$F_0 \lleftarrow \bm{\zeroexp}$;\quad $F_1 \lleftarrow \oneexp$;\quad $k \lleftarrow 1$; \quad \changed{$\oldSubsystem \lleftarrow \emptyset$\;} \label{alg:pric3:init}
\While(\label{alg:pric3:loop}){$\true$}
{
$\success, \,  F_0, \ldots,F_k, \, \subsystem  \lleftarrow \strengthencall{\frameseqk}{\heuristic}$\; \label{alg:pric3:strengthen}
    \lIf(\label{alg:pric3:unknown}){$\neg\success$}{\Return \false, \subsystem}
       $F_{k+1} \lleftarrow \oneexp$\; \label{alg:pric3:newFrame}
       $F_0, \ldots,F_{k+1} \lleftarrow \propagatecall{F_0,\ldots,F_{k+1}}$\; \label{alg:pric3:propagate}
       \lIf(\label{alg:pric3:checkInductive}){$\exists\, 1 \leq i \leq k\colon F_i = F_{i+1}$}
       {\Return \true, \blank}
       \changed{\lIf(\label{alg:pric3:zeno}){$\oldSubsystem = \subsystem$}{\Return \false, \subsystem}}
       $k \lleftarrow k+1$; \quad \changed{$\oldSubsystem \lleftarrow \subsystem$}\;
    }
    \caption{$\prictcall{\mdp}{\bad}{\thresh}{\heuristic}$}
    \label{alg:pric3}
\end{algorithm}%
%
Analogously to standard \icthree, this approach is taken by the core loop \mainloop depicted in Alg.~\ref{alg:pric3}; differences to the main loop of \icthree (cf.~\cite[Fig. 5]{DBLP:conf/fmcad/EenMB11}) are highlighted in \changed{\changedColor}.
A~particular difference is that \mainloop is parameterized by a heuristic \heuristic for finding suitable probabilities (see Chall.~\ref{challenge:strengthen}). 
Since the precise choice of \heuristic is irrelevant for the 
soundness 
of \mainloop, we defer a detailed discussion of suitable heuristics to Sec.~\ref{sec:strengthen}.

As input, \mainloop takes an MDP $\mdp = \mdps$, a set $\bad \subseteq S$ of bad states, and a threshold $\thresh \in [0,1]$. 
Since the input is never changed, we assume it to be \emph{globally available}, also to subroutines.
As output, \mainloop returns $\true$ if two consecutive frames become equal.
We hence say that \mainloop is \emph{sound} if it only returns $\true$ if $\mdp$ is safe.

We will formalize soundness using Hoare triples. For precondition $\phi$, postcondition~$\psi$, and program $P$, the triple $\htriple{\phi}{P}{\psi}$ is \emph{valid} (for partial correctness) if, whenever program $P$ starts in a state satisfying precondition $\phi$ and terminates in some state $s'$, then $s'$ satisfies postcondition $\psi$.
Soundness of \mainloop then means validity of the triple%
\begin{align*}
	\htriple{\true}{\safeReturn, \blank \lleftarrow \prictcall{\mdp}{\bad}{\thresh}{\heuristic}}{\safeReturn \RRightarrow \prevtmax{\mc}{\sinit}{\bad} \leq \thresh }~.
\end{align*}%
Let us briefly go through the individual steps of \mainloop in Alg.~\ref{alg:pric3} and convince ourselves that it is indeed sound.
After that, we discuss why \mainloop terminates and what happens if it is unable to prove safety by finding two equal consecutive frames.%

\subsubsection{How \mainloop works.}
Recall that \mainloop maintains a sequence of frames $\frameseqk$ which is initialized in l.~\ref{alg:pric3:init} with $k = 1$, $F_0 = \bm{\zeroexp}$, and $F_1 = \oneexp$, where the frame $\oneexp$ maps every state to $1$.
Every time upon entering the \textbf{while}-loop in terms l.~\ref{alg:pric3:loop}, the initial segment $\frameseqkmo$ satisfies all \pricthree invariants (cf.~Def.~\ref{def:invariants}), whereas the full sequence $\frameseqk$ potentially violates frame-safety as it is possible that $F_{k}[\sinit] > \thresh$.

In l.~\ref{alg:pric3:strengthen}, procedure $\strengthen$---detailed in Sect.~\ref{sec:strengthen}---is called to restore \emph{all} \pricthree invariants on the \emph{entire} frame sequence:
It either returns $\true$ if successful
or returns $\false$ and a counterexample (in our case a subsystem of the MDP) if \mbox{it was unable to do so}.
To ensure soundness of \mainloop, it suffices that $\strengthen$ restores the \pricthree invariants whenever it returns $\true$. Formally, $\strengthen$ must meet the following specification:%
\begin{definition}
\label{def:spec_strengthen}
	Procedure $\strengthen$ is \emph{sound} if the following Hoare triple is valid:
	\begin{align*}
		& 	\big\{\, 
			\icinv{F_0,\ldots,F_{k-1}} 
			\eewedge 
			F_{k-1} \leq F_k 
			\eewedge 
			\bm{F_{k-1}} \leq F_k 
			\,\big\}\\
	        & \qquad \success , \,  F_0, \ldots,F_k, \, \blank \lleftarrow  \strengthencall{\frameseqk}{\heuristic} \\
			& \big\{\, \success \RRightarrow \icinv{F_0, \ldots, F_k} 
			\,\big\}. 
	\end{align*}%
\end{definition}%
%
%
If $\strengthen$ returns $\true$, then a new frame $F_{k+1} = \oneexp$ is created in l.~\ref{alg:pric3:newFrame}.
After that,  the (now initial) segment $\frameseqk$ again satisfies all \pricthree invariants, whereas the full sequence $\frameseqkpo$ potentially violates frame-safety at $F_{k+1}$.
\emph{Propagation}~(l.~\ref{alg:pric3:propagate}) aims to speed up termination by updating $F_{i+1}[s]$ by $F_{i}[s]$ iff this does not violate relative inductivity.
Consequently, the previously mentioned properties remain unchanged.

If $\strengthen$ returns $\false$, the \pricthree invariants---premises to Lem.~\ref{lem:ic3_inv_eqframes} for witnessing safety---cannot be restored and \mainloop terminates returning $\false$ (l.~\ref{alg:pric3:unknown}).
Returning $\false$ (also possible in l.~\ref{alg:pric3:zeno}) has by specification no affect on soundness of \mainloop .

In l.~\ref{alg:pric3:checkInductive}, we check whether there exist two identical consecutive frames.
If so, Lem.~\ref{lem:ic3_inv_eqframes} yields that the MDP is safe; consequently, \mainloop returns $\true$.
Otherwise, we increment~$k$ and are in the same setting as upon entering the loop, now with an increased frame sequence; \mainloop then performs another iteration. In summary, we obtain:%
\begin{theorem}[Soundness of $\boldmainloop$]
\label{thm:strengthencorrecticcorrect}
   If $\strengthen$ is sound
   and $\propagate$ does not affect the \pricthree invariants, 
   then \mainloop is sound, i.e., the following triple is valid:
\begin{align*}
\htriple{\true}{\safeReturn, \blank \lleftarrow \prictcall{\mdp}{\bad}{\thresh}{\heuristic}}{\safeReturn \iimplies \prevtmax{\mc}{\sinit}{\bad} \lleq \thresh }
\end{align*}%
\end{theorem}%

\vspace{-1ex}\subsubsection{\mainloop terminates for unsafe MDPs.}
If the MDP is unsafe, then there exists a step-bound $n$, such that
  $\prevtbmax{\mdp}{\sinit}{n}{\bad} > \thresh$.
Furthermore, any sound implementation of \strengthen{} (cf.\ Def.~\ref{def:spec_strengthen}) either immediately terminates \mainloop by returning $\false$ or restores the \pricthree invariants for $\frameseqk$.
If the former case never arises, then \strengthen{} will eventually restore the \pricthree invariants for a frame sequence of length $k = n$.
By Lem.~\ref{lem:ic3_inv_overapprox}, we have 
$F_n[\sinit] \geq \prevtbmax{\mdp}{\sinit}{n}{\bad} > \thresh$ contradicting frame-safety.

\subsubsection{\mainloop terminates for safe MDPs.}
Standard \icthree terminates on safe finite TSs as there are only finitely many different frames, making every ascending chain of frames eventually stabilize. 
For us, frames map states to probabilities (Chall.~\ref{challenge:domain}), yielding \emph{infinitely many possible frames} even for finite MDPs. 
Hence, $\strengthen$ need not ever yield a stabilizing chain of frames.
If it continuously fails to stabilize while repeatedly reasoning about the same set of states, we give up.
\mainloop checks this by comparing the subsystem $\strengthen$ operates on with the one it operated on in the previous loop iteration (l.~\ref{alg:pric3:zeno}).
\begin{theorem}
If \strengthen{} and \propagate terminate, then \mainloop terminates.	
\end{theorem}
\begin{recovery}
	\label{recovery:2}
For qual.~reachability ($\lambda=0$), \mainloop never terminates in \textnormal{l.~\ref{alg:pric3:zeno}}.
\end{recovery}


\vspace{-1.4ex}\subsubsection{\mainloop is incomplete.}
Standard \icthree either proves safety 
or returns $\false$ and a counterexample---a single path from the initial to a bad state.
As single paths are insufficient as counterexamples in MDPs~(Chall.~\ref{challenge:counterexamples}), \mainloop instead returns a \emph{subsystem} of the MDP $\mdp$ provided by \strengthen.
However, as argued above, we cannot trust $\strengthen$ to provide a stabilizing chain of frames. 
Reporting $\false$ thus only means that the given MDP \emph{may} be unsafe; the returned subsystem has to be analyzed further.

%
The full \pricthree algorithm presented in Sect.~\ref{sec:outermost_loop} addresses this issue.
Exploiting the subsystem returned by \mainloop, \pricthree returns $\true$ if the MDP is safe; otherwise, it returns $\false$ and provides a true counterexample witnessing that the MDP is unsafe.
\begin{figure}[t]
\subfloat[Threshold $\lambda = \sfrac{5}{9}$]{
\adjustbox{max width=0.56\textwidth}{
\begin{tabular}{l|c|cc|ccc|cccc|ccccc}
It. & 
1   & 
\multicolumn{2}{c|}{2} & 
\multicolumn{3}{c|}{3} & 
\multicolumn{4}{c|}{4}  &
\multicolumn{5}{c}{5} \\\hline
$F_i$ & 
$F_1$ & 
$F_1$ & 
$F_2$ & 
$F_1$ & 
$F_2$ & 
$F_3$ & 
$F_1$ & 
$F_2$ & 
$F_3$ & 
$F_4$ & 
$F_1$ & 
$F_2$ & 
$F_3$ & 
$F_4$ & 
$F_5$ \\
\hline
$s_0$ & 
$\sfrac{5}{9}$ &
$\sfrac{5}{9}$ & 
$\sfrac{5}{9}$ & 
$\sfrac{5}{9}$ & 
$\sfrac{5}{9}$ & 
$\sfrac{5}{9}$ & 
$\sfrac{5}{9}$ & 
$\sfrac{5}{9}$ & 
$\sfrac{5}{9}$ & 
$\sfrac{5}{9}$ & 
$\sfrac{5}{9}$ & 
$\sfrac{5}{9}$ & 
$\sfrac{5}{9}$ & 
$\sfrac{5}{9}$ & 
$\sfrac{5}{9}$ \\
$s_1$ & 
\old{$1$} & 
$\sfrac{11}{18}$ & 
\old{$1$} & 
$\sfrac{11}{18}$ & 
$\sfrac{11}{18}$ &  
\old{$1$} & 
$\sfrac{11}{18}$ & 
$\sfrac{11}{18}$ & 
$\sfrac{11}{18}$ & 
\old{$1$} & 
$\sfrac{11}{18}$ & 
$\sfrac{11}{18}$ & 
$\sfrac{11}{18}$ & 
$\sfrac{11}{18}$ & 
\old{$1$} \\
$s_2$ & 
\old{$1$} & 
$\sfrac{1}{2}$ &
\old{$1$} & 
$\sfrac{1}{2}$ & 
$\sfrac{1}{2}$ & 
\old{$1$} & 
$\sfrac{1}{2}$ &
$\sfrac{1}{2}$ & 
$\sfrac{1}{2}$ & 
\old{$1$} & 
$\sfrac{1}{2}$ &
$\sfrac{1}{2}$ & 
$\sfrac{1}{2}$ & 
$\sfrac{1}{2}$ & 
\old{$1$} \\
$s_3$ & 
\old{$1$} & 
\old{$1$} & 
\old{$1$} & 
$\sfrac{2}{3}$ & 
\old{$1$} & 
\old{$1$} & 
$\sfrac{2}{3}$ &
$\sfrac{2}{3}$ & 
\old{$1$} &
\old{$1$} &
$\sfrac{2}{3}$ &
$\sfrac{2}{3}$ &
$\sfrac{2}{3}$ & 
\old{$1$} &
\old{$1$} \\
$s_4$ & 
\old{$1$} &
\old{$1$} &
\old{$1$} &
\old{$1$} &
\old{$1$} &
\old{$1$} &
$0$ & 
\old{$1$} & 
\old{$1$} & 
\old{$1$} & 
$0$ & 
$0$ & 
\old{$1$} & 
\old{$1$} & 
\old{$1$} \\
$s_5$ & 
\old{$1$} &
\old{$1$} &
\old{$1$} &
\old{$1$} &
\old{$1$} &
\old{$1$} &
\old{$1$} &
\old{$1$} &
\old{$1$} &
\old{$1$} &
\old{$1$} &
\old{$1$} &
\old{$1$} &
\old{$1$} &
\old{$1$} 
\end{tabular}
}
}
\subfloat[Threshold $\lambda = \sfrac{9}{10}$]{
\quad
\adjustbox{max width=0.405\textwidth}{
\begin{tabular}{l|c|cc|ccc|cccc}
It. & 
1 & 
\multicolumn{2}{c|}{2} & 
\multicolumn{3}{c|}{3} &
\multicolumn{4}{c}{4}  \\\hline  
$F_i$ & 
$F_1$ & 
$F_1$ & 
$F_2$ & 
$F_1$ & 
$F_2$ & 
$F_3$ &
$F_1$ & 
$F_2$ & 
$F_3$ &
$F_4$ \\\hline
$s_0$ & 
$\sfrac{9}{10}$ & 
$\sfrac{9}{10}$ & 
$\sfrac{9}{10}$ & 
$\sfrac{9}{10}$ & 
$\sfrac{9}{10}$ & 
$\sfrac{9}{10}$ & 
$\sfrac{9}{10}$ & 
$\sfrac{9}{10}$ & 
$\sfrac{9}{10}$ & 
$\sfrac{9}{10}$ \\
$s_1$ & 
\old{$1$} & 
$\sfrac{99}{100}$ & 
\old{$1$} &
$\sfrac{99}{100}$ & 
$\sfrac{99}{100}$ & 
\old{$1$} &
$\sfrac{99}{100}$ & 
$\sfrac{99}{100}$ & 
$\sfrac{99}{100}$ & 
\old{$1$} \\
$s_2$ &
\old{$1$} &
$\sfrac{81}{100}$ &
\old{$1$} &
$\sfrac{81}{100}$ &
$\sfrac{81}{100}$ &
\old{$1$} &
$\sfrac{81}{100}$ &
$\sfrac{81}{100}$ &
$\sfrac{81}{100}$ &
\old{$1$} \\
$s_3$ &
\old{$1$} &
\old{$1$} &
\old{$1$} &
\old{$1$} &
\old{$1$} &
\old{$1$} &
\old{$1$} &
\old{$1$} &
\old{$1$} &
\old{$1$} \\
$s_4$ &
\old{$1$} &
\old{$1$} &
\old{$1$} &
$0$ &
\old{$1$} &
\old{$1$} &
$0$ &
$0$ &
\old{$1$} &
\old{$1$} \\
$s_5$ &
\old{$1$} &
\old{$1$} &
\old{$1$} &
\old{$1$} &
\old{$1$} &
\old{$1$} &
\old{$1$} &
\old{$1$} &
\old{$1$} &
\old{$1$}
\end{tabular}
}
}
\caption{Two runs of \mainloop on the Markov chain induced by selecting action $a$ in Fig.~\ref{fig:runningex}. For every iteration, frames are recorded after invocation of $\strengthen$.}
\label{fig:naiveic3frames}
\end{figure}%

\begin{example}
\label{ex:running:setup}
We conclude this section with two example executions of $\mainloop$ on a simplified version of the MDP in Fig.~\ref{fig:runningex}.
Assume that action $b$ has been removed. Then, for every state, exactly one action is enabled, i.e., we consider a Markov chain.
Fig.~\ref{fig:naiveic3frames} depicts the frame sequences computed by \mainloop (for a reasonable $\heuristic$) on that Markov chain for two thresholds: 
$\nicefrac{5}{9} \textcolor{gray}{{}= \prevtmax{\mc}{s_0}{\bad}}$ and $\nicefrac{9}{10}$.
In particular, notice that \emph{proving the coarser bound of $\nicefrac{9}{10}$
requires fewer frames than proving the exact bound of $\nicefrac{5}{9}$}.
\hfill$\triangle$
\end{example}%
\section{Strengthening in \mainloop}
\label{sec:strengthen}

%
%

When the main loop of \mainloop has created a new frame $F_{k} = \oneexp$ in its previous iteration, this frame may violate frame-safety (Def.~\ref{def:invariants}.3) because of $F_{k}[\sinit] = 1 \not\leq \lambda$.
The task of $\strengthen$ is to restore the \pricthree invariants on \emph{all} frames $F_0,\hdots,F_k$.
To this end, our first \emph{obligation} is to lower the value in frame $i = k$ for state $s = \sinit$ to $\delta = \lambda \in [0,1]$. 
We denote such an obligation by $(i,s,\delta)$.
Observe that implicitly $\delta = 0$ in the qualitative case, i.e., when proving unreachability. 
An obligation $(i,s,\delta)$ is \emph{resolved} by updating the values assigned to state $s$ in \emph{all frames} $F_1, \ldots, F_i$ to at most $\delta$.
That is, for all $j \leq i$, we set $F_{j}[s]$ to the minimum of $\delta$ and the original value $F_{j}[s]$.
Such an update affects neither initiality nor the chain property (Defs.~\ref{def:invariants}.1,~\ref{def:invariants}.2).
It may, however, violate relative inductivity (Def.~\ref{def:invariants}.4), i.e., $\bm{F_{i-1}} \leq F_{i}$. 
Before resolving obligation $(i,s,\delta)$, we may thus have to further decrease some entries in $F_{i-1}$ as well.
Hence, \emph{resolving obligations may spawn additional obligations} which have to be resolved first to maintain relative inductivity.
In this section, we present a generic instance of $\strengthen$
meeting its specification (Def.~\ref{def:spec_strengthen}) and discuss its correctness.

\label{sec:strengthenexplained}

\subsubsection{$\strengthen$ by example.}
$\strengthen$ is given by the pseudo code in Alg.~\ref{alg:strengthensimple}; differences to standard \icthree
(cf.~\cite[Fig. 6]{DBLP:conf/fmcad/EenMB11}) are highlighted in \changed{\changedColor}.
Intuitively, $\strengthen$ attempts to recursively resolve all obligations until either both frame-safety and relative inductivity are restored for \emph{all} frames 
or it detects a \emph{potential counterexample} justifying why it is unable to do so.
We first consider an execution where the latter does not arise:
%
\definecolor{webgreen}{rgb}{0,.5,0}
\newcommand\mycommfont[1]{\ttfamily\textcolor{teal}{#1}}
\SetCommentSty{mycommfont}
\begin{algorithm}[t]
\SetNoFillComment
$Q \leftarrow \{(k, \sinit, \changed{\lambda})\}$ \;\label{alg:str:init}
\While{$Q$ not empty}{
        $(i, s, \changed{\delta}) \leftarrow Q.\qpop$\tcc*{pop obligation with minimal frame index} \label{alg:str:pop}
        \uIf(\label{alg:str:frame-zero}){$i = 0\quad\changed{\vee\quad(s \in \bad \wedge \delta < 1)}$}{
       \tcc{\changed{possible} counterexample given by subsystem consisting of states popped from Q at some point}
       \Return $\false$\changed{, \blank, \qtouched};
       \label{alg:str:false}
          
   }
   \tcc{check whether $F_i[s] \leftarrow \changed{\delta}$ violates relative inductivity}
   \uIf(\label{alg:str:relind}for such an $a$){$\exists a \in \acts{s} \colon \bma{a}{F_{i-1}}[s] > \changed{\delta}$}{
      	%
      	%
      	%
      %
      %
          %
          %
           $\changed{\delta_1,\ldots,\delta_n \leftarrow \getprobs\left( s, \, a, \, \delta \right)}$ \;\label{alg:str:getprobs}
          %
          %
                  $\{s_1, \ldots, s_n\} \leftarrow \succs(s, a)$\;\label{alg:str:succs}
                  $Q.\qpush{\left(i-1, s_1,  \changed{\delta_1} \right), \ldots, \left(i-1, s_n,  \changed{\delta_n} \right), \left( i, s, \changed{\delta} \right) }$\;\label{alg:str:push}
             %
      %
             %
      %
      }
      \uElse(\tcc*[h]{resolve $(i,s,\delta)$ without violating relative inductivity})
      {
       %
       %
       	%
       %
       %
	   %
       %
       $F_1[s] \leftarrow \minval{F_1[s], \, \changed{\delta}}; \ldots ; F_i[s]  \leftarrow \minval{F_i[s], \, \changed{\delta}}$\; \label{alg:str:update}
      }
}(\tcc*[h]{$Q$ empty; all obligations have been resolved} \label{alg:str:endwhile})
\Return $\true, \, F_0,\ldots ,F_k, \changed{\qtouched}$\;
\label{alg:str:true}
   
   \caption{$\strengthencall{F_0, \ldots,F_k}{\heuristic}{}$}
   \label{alg:strengthensimple}
\end{algorithm}%
%
%
%
\begin{example}
\label{ex:running:strengthen}
  We zoom in on Ex.~\ref{ex:running:setup}:
  Prior to the second iteration, we have created the following three frames assigning values to the states $s_0,s_5$:
  \begin{align*}
          F_0 \eeq (0,0,0,0,1), 
    \qquad
    F_1 \eeq (\sfrac{5}{9}, 1,1,1,1,1),\qquad \text{and}
    \qquad
    F_2 \eeq \oneexp.
  \end{align*}
  To keep track of unresolved obligations $(i,s,\delta)$, 
  $\strengthen$ employs a priority queue~$Q$ which pops obligations with minimal frame index $i$ first.
  Our first step is to ensure frame-safety of $F_2$, i.e., alter $F_2$ so that $F_2[s_0] \leq \sfrac{5}{9}$; we thus initialize the queue $Q$ with the initial obligation $(2,s_0,\sfrac{5}{9})$ (l.~\ref{alg:str:init}).
  To do so, we check whether updating $F_2[s_0]$ to $\sfrac{5}{9}$ would invalidate relative inductivity (l.~\ref{alg:str:relind}).
  This is indeed the case:
  \begin{align*}
    \bm{F_1}[s_0] 
    \eeq 
    \sfrac{1}{2} \cdot F_1[s_1] + \sfrac{1}{2} \cdot F_1[s_2] 
    \eeq
    1 \not\leq \sfrac{5}{9}.
  \end{align*}
  To restore relative inductivity, $\strengthen$ spawns one new obligation for each relevant successor of $s_0$.
  These have to be resolved before retrying to resolve the old obligation.\footnote{We assume that the set $\succs(s,a) = \left\{ s' \in S ~\mid~ P(s, a , s') > 0 \right\}$ of \emph{relevant $a$-successors} of state $s$ is returned in some arbitrary, but fixed order.}
  
  \emph{In contrast to standard \icthree, spawning obligations involves finding suitable probabilities $\delta$} (l.~\ref{alg:str:getprobs}).
  In our example this means we have to spawn two obligations $(1,s_1,\delta_1)$ and $(1,s_2,\delta_2)$ such that $\sfrac{1}{2} \cdot \delta_1 + \sfrac{1}{2} \cdot \delta_2 \leq \sfrac{5}{9}$.
There are \emph{infinitely many choices} for $\delta_1$ and $\delta_2$ satisfying  this inequality.
  Assume some heuristic $\getprobs$ chooses $\delta_1 = \sfrac{11}{18}$ and $\delta_2 = \sfrac{1}{2}$; we push obligations $(1,s_1,\sfrac{11}{18})$, $(1,s_2,\sfrac{1}{2})$, and $(2,s_0,\sfrac{5}{9})$~(ll.~\ref{alg:str:succs},~\ref{alg:str:push}).
  In the next iteration, we first pop obligation $(1,s_1,\sfrac{11}{18})$ (l.~\ref{alg:str:pop}) and find that it can be resolved without violating relative inductivity  (l.~\ref{alg:str:relind}).
  Hence, we set $F_1[s_1]$ to $\sfrac{11}{18}$ (l.~\ref{alg:str:update}); no new obligation is spawned.
  Obligation $(1,s_2,\sfrac{1}{2})$ is resolved analogously; the updated frame is $F_1 = (\sfrac{5}{9},\sfrac{11}{18},\sfrac{1}{2},1)$.
  Thereafter, our initial obligation $(2,s_0,\sfrac{5}{9})$ can be resolved; relative inductivity is restored for $F_0,F_1,F_2$.
  Hence, $\strengthen$ returns $\true$ together with the updated frames. 
    \hfill $\triangle$
\end{example}
%

%
\vspace{-1ex}\subsubsection{$\strengthen$ is sound.}
Let us briefly discuss why Alg.~\ref{alg:strengthensimple} meets the specification of a sound implemenation of $\strengthen$ (Def.~\ref{def:spec_strengthen}):
First, we observe that Alg.~\ref{alg:strengthensimple} alters the frames---and thus potentially invalidates the \pricthree invariants---only in l.~\ref{alg:str:update} by resolving an obligation $(i,s,\delta)$ with $\bm{F_{i-1}}[s] \leq \delta$ (due to the check in l.~\ref{alg:str:relind}).

\noindent
Let $F\subst{s}{\delta}$ denote the frame $F$ in which $F[s]$ is set to $\delta$, i.e.,%
\begin{align*}
	F\subst{s}{\delta}[s'] \eeq  
	\begin{cases}
		\delta, &\text{if}~s' = s, \\
		F[s'], &\text{otherwise}.
	\end{cases}
\end{align*}%
Indeed, resolving obligation $(i,s,\delta)$ in l.~\ref{alg:str:update} lowers the values assigned to state $s$ to at~most~$\delta$ \emph{without} invalidating the \pricthree invariants:
\begin{lemma}\label{lem:ic3_inv_preserved_by_update}
Let $(i,s,\delta)$ be an obligation and $F_0,\ldots,F_i$, for $i > 0$, be frames with $\bm{F_{i-1}}[s] \leq \delta$. 
Then $\icinv{F_0, \ldots, F_i} \text{ implies }$
\begin{align*}
\icinv{~\vphantom{\Bigl(}
F_0 \subst{s}{\vphantom{\bigl(}\minval{F_0[s],\, \delta}},
\,\ldots,\,
F_i \subst{s}{\vphantom{\bigl(}\minval{F_i[s],\, \delta}}~}.
\end{align*}
\end{lemma}%
Crucially, the precondition of Def.~\ref{def:spec_strengthen} guarantees that all \pricthree invariants except frame safety hold initially.
Since these invariants are never invalidated due to Lem.~\ref{lem:ic3_inv_preserved_by_update}, Alg.~\ref{alg:strengthensimple} is a sound implementation of $\strengthen$ if it restores frame safety whenever it returns $\true$, i.e., once it leaves the loop with an empty obligation queue $Q$ (ll.~\ref{alg:str:endwhile}--\ref{alg:str:true}).
Now, an obligation $(i,s,\delta)$ is only popped from $Q$ in l.~\ref{alg:str:pop}. 
As $(i,s,\delta)$ is added to $Q$ upon reaching l.~\ref{alg:str:push}, the size of $Q$ can only ever be reduced (without returning $\false$) by resolving $(i,s,\delta)$ in l.~\ref{alg:str:update}.
Hence, Alg.~\ref{alg:strengthensimple} does not return $\true$ unless it restored frame safety by resolving, amongst all other obligations, the initial obligation
$(k,\sinit,\thresh)$. 
Consequently:
\begin{lemma}
\label{lem:strenghtencorrect}
  Procedure $\strengthen$ is sound, i.e., it satisfies the specification in \textnormal{Def.~\ref{def:spec_strengthen}}.
\end{lemma}
\begin{theorem}
\label{thm:priccorrect}
Procedure $\mainloop$ is sound, i.e., satisfies the specification in \textnormal{Thm.~\ref{thm:strengthencorrecticcorrect}}.
\end{theorem}
%
We remark that, analogously to standard \icthree, resolving an obligation 
in l.~\ref{alg:str:update} may be accompanied by \emph{generalization}. 
That is, we attempt to update the values of multiple states at once.
Generalization is, however, highly non-trivial in a probabilistic setting.
We discuss three possible approaches to generalization in Sect.~\ref{sec:generalization}.

\subsubsection{\strengthen{} terminates.}
\label{sec:strengthen:getprobs}
We now show that $\strengthen$ as in Alg.~\ref{alg:strengthensimple} terminates.
The only scenario in which $\strengthen$ may not terminate is if it keeps spawning obligations in l.~\ref{alg:str:push}. 
Let us thus look closer at how obligations are spawned:
%
Whenever we detect that resolving an obligation $(i,s,\delta)$ would violate relative inductivity for some action $a$ (l.~\ref{alg:str:relind}), we first need to update the values of the successor states $s_1,\hdots, s_n \in \succs(s,a)$ in frame $i{-}1$, i.e., we push the obligations $(i{-}1,s_1,\delta_1), \hdots, (i{-}1, s_n, \delta_n)$ which have to be resolved first (ll.~\ref{alg:str:getprobs}--\ref{alg:str:push}).
It is noteworthy that, for a TS, a single action leads to a single successor state $s_1$.
Alg.~\ref{alg:strengthensimple} employs a heuristic \heuristic to determine the probabilities required for pushing obligations (l.~\ref{alg:str:getprobs}).
Assume for an obligation $(i,s,\delta)$ that the check in l.~\ref{alg:str:relind} yields 
$\exists a \in \acts{s} \colon \bma{a}{F_{i-1}}[s] > \delta$. 
Then \heuristic takes $s$, $a$, $\delta$ and reports some probability $\delta_j$ for every $a$-successor $s_j$ of $s$.
However, an arbitrary heuristic of type
$\heuristic\colon S \,\times\, \act \,\times\, [0,1] \to [0,1]^*$ may lead to non-terminating behavior:
If $\delta_1, \hdots, \delta_n = F_{i-1}[s_1], \hdots F_{i-1}[s_n]$, then the heuristic has no effect. 
It is thus natural to require that an \emph{adequate} heuristic \heuristic yields probabilities such that the check $\bma{a}{F_{i-1}}[s] > \delta$ in l.~\ref{alg:str:relind} cannot succeed twice for the \emph{same obligation} $(i,s,\delta)$ and \emph{same action $a$}. Formally, this is guaranteed by the following: 
%
%
\begin{definition}
\label{descr:goodheuristic} 
Heuristic \heuristic is \emph{adequate} if the following triple is valid (for any frame $F$):
\begin{align*}
  & \bigl\{\, \succs(s,a) = s_1, \ldots, s_n \,\bigr\} \\
  & \quad \delta_1,\ldots,\delta_n \lleftarrow \heuristic(s,a,\delta) \\
  & \bigl\{\, \bma{a}{\vphantom{\big(}F\subst{s_1}{\delta_1}\ldots\subst{s_n}{\delta_n}}[s] \leq \delta \,\bigr\}
  \tag*{$\triangle$}
\end{align*}
\end{definition}%
Details regarding our implementation of heuristic \heuristic are found in Section~\ref{sec:ic3instantiated}.

For an adequate heuristic, attempting to resolve an obligation $(i,s,\delta)$ (ll.~\ref{alg:str:pop} -- \ref{alg:str:update}) either succeeds after spawning it at most $|\act(s)|$ times or $\strengthen$ returns $\false$.
By a similar argument, attempting to resolve an obligation $(i > 0 , s, \blank\,)$ leads to at most $\sum_{a\in\act(s)} |\{ s' \in S \mid P(s, a, s') > 0\}|$ other obligations of the form $(i{-}1, s', \blank\,)$.
Consequently, the total number of obligations spawned by Alg.~\ref{alg:strengthensimple} is bounded. 
Since Alg.~\ref{alg:strengthensimple} terminates if all obligations have been resolved (l.~\ref{alg:str:endwhile}) and each of its loop iterations either returns $\false$, spawns obligations, or resolves an obligation, we conclude:
\begin{lemma}
$\strengthen(\frameseqk)$ terminates for every adequate heuristic $\heuristic$.
\end{lemma}%
\begin{recovery}
	\label{recovery:3}
Let \heuristic be adequate. 
Then for qualitative reachability ($\thresh = 0$), all obligations spawned by $\strengthen$ as in \textnormal{Alg.~\ref{alg:strengthensimple}} are of the form $(i,s,0)$.
%
\end{recovery}%
%
\vspace{-1ex}\subsubsection{\strengthen{} returns $\false$.}
\label{sec:strengthen:counterexamples}
There are two cases in which $\strengthen$ fails to restore the \pricthree invariants and returns $\false$.
The first case (the left disjunct of l.~\ref{alg:str:frame-zero}) is that we encounter an obligation for frame $F_0$.
Resolving such an obligation would inevitably violate \emph{initiality}; analogously to standard \icthree, we thus return $\false$.

The second case~(the right disjunct of l.~\ref{alg:str:frame-zero}) is that we encounter an obligation $(i,s,\delta)$ for a bad state $s \in \bad$ with a probability $\delta < 1$ (though, obviously, all $s \in \bad$ have probability ${=}1$).
Resolving such an obligation would inevitably prevents us from restoring \emph{relative inductivity}: 
If we updated $F_i[s]$ to $\delta$, we would have $\bm{F_{i-1}}[s] = 1 > \delta = F_{i}[s]$. 
Notice that, in contrast to standard \icthree, this second case \emph{can} occur in \pricthree:
\begin{example}
\label{ex:heuristic-no-values}
Assume we have to resolve an obligation $(i, s_3, \sfrac{1}{2})$ for the MDP in Fig.~\ref{fig:runningex}.
This involves spawning obligations $(i{-}1, s_4,\delta_1)$ and $(i{-}1, s_5, \delta_2)$, where $s_5$ is a bad state, such that
$\sfrac{1}{3} \cdot \delta_1 + \sfrac{2}{3} \cdot \delta_2 \leq \sfrac{1}{2}$.
Even for $\delta_1 = 0$, this is only possible if $\delta_2 \leq \sfrac{3}{4} < 1$.
\hfill$\triangle$
\end{example}

\vspace{-1ex}\subsubsection{\strengthen{} cannot prove unsafety.}
\label{sec:strengthen:counterexamples}
If standard \icthree returns $\false$, it proves unsafety by constructing a counterexample, i.e., \emph{a single path from the initial state to a bad state}. 
If \pricthree returns $\false$, there are two possible reasons: 
\emph{Either} the MDP is indeed unsafe, \emph{or} the heuristic $\heuristic$ at some point selected probabilities in a way such that $\strengthen$ is unable to restore the $\pricthree$ invariants (even though the MDP might in fact be safe).
\strengthen{} thus only returns a \emph{potential} counterexample which either proves unsafety or indicates that our heuristic was inappropriate.

Counterexamples in our case consist of subsystems rather than a single path (see Chall.~\ref{challenge:counterexamples} and Sec.~\ref{sec:outermost_loop}).
\strengthen{} hence returns the set $\qtouched$ of all states that eventually appeared in the obligation queue. 
This set is a conservative approximation, and optimizations as in~\cite{DBLP:conf/sfm/AbrahamBDJKW14} may be beneficial.
Furthermore, in the qualitative case, our potential counterexample subsumes the counterexamples constructed by standard \icthree:
\begin{recovery}
\label{recovery:4}
    Let $\heuristic_0$ be the adequate heuristic mapping every state to $0$.
    For qual.\ reachability ($\lambda = 0$), 
	if $\success = \false$ is returned by $\strengthencall{\frameseqk}{\heuristic_0}$, then
	$\qtouched$ contains a path from the initial to a bad state.\textnormal{\footnote{$\qtouched$ might be restricted to only contain this path by some simple adaptions.}}
\end{recovery}

\section{Dealing with Potential Counterexamples}
\label{sec:outermost_loop}



Recall that our core algorithm $\mainloop$ is incomplete for a fixed heuristic~\heuristic:
It cannot give a conclusive answer whenever it finds a potential counterexample for two possible reasons:
Either the heuristic $\heuristic$ turned out to be inappropriate 
or the MDP is indeed unsafe.	
The idea to overcome the former is to call $\mainloop$ finitely often in an outer loop that generates new heuristics until we find an appropriate one: 
If \mainloop still does not report safety of the MDP, then it is indeed unsafe.
We do not blindly generate new heuristics, but use the potential counterexamples returned by $\mainloop$ to refine the previous one.

%
%
\begin{algorithm}[t]
\KwData{global MDP $\mdp$,\quad set of bad states $\bad$,\quad  threshold $\lambda$ }
\KwResult{$\true$ iff $\prevtmax{\mdp}{\sinit}{\bad} \leq \lambda$}
$\oracle  \lleftarrow \interface{\initOracle()}$;\label{alg:outermost:init}
$\touched \lleftarrow \{\sinit\}$\;
\Do{
    $\touched \neq S$
}{
    $\heuristic \lleftarrow \createHeuristic(\oracle)$; $\quad\safeReturn, \subsystem \lleftarrow \mainloop()$\;

        \lIf{
      \safeReturn
    }{
      \Return $\true$
    }
    \lIf(\label{alg:outer:refute}){
            $\checkrefutation(\subsystem)$
    }{
        \Return $\false$
    }
    {$\touched \lleftarrow \interface{\enlarge(\touched, \subsystem)}$\;\label{alg:outermost:touched}}
    $\Omega \lleftarrow \interface{\refineoracle(\Omega, \touched)}$\;\label{alg:outermost:refine}
}
return $\Omega(\sinit) \leq \lambda$
    \caption{\pricthree: The outermost loop dealing with possibly imprecise heuristics}
    \label{alg:outermost_loop}
\end{algorithm}%
%
%


Let consider the procedure $\outerloop$ in Alg.~\ref{alg:outermost_loop} which wraps our core algorithm \mainloop in more detail:
First, we create an \emph{oracle} $\Omega\colon S \rightarrow [0,1]$ which (roughly) \emph{estimates} the probability of reaching $\bad$ for every state.
A~\emph{perfect oracle} would yield \emph{precise} maximal reachability probabilites, i.e., $\Omega(s) = \prevtmax{}{s}{\bad}$ for every state $s$.
We construct oracles by \interface{user-supplied methods} (highlighted in \interface{blue}). 
Examples of implementations of all user-supplied methods in Alg.~\ref{alg:outermost_loop} are discussed in Sect.~\ref{sec:experiments}.

Assuming the oracle is good, but not perfect, we construct an adequate heuristic~$\heuristic$ selecting probabilities based on the oracle\footnote{We thus assume that heuristic \heuristic invokes the oracle whenever it needs to guess some probability.} for all successors of a given state: There are various options. 
The simplest is to pass-through the oracle values. 
A version that is more robust against noise in the oracle is discussed in Sect.~\ref{sec:practical}.
We then invoke $\mainloop$.
If $\mainloop$ reports safety, the MDP is indeed safe by the soundness of $\mainloop$.

\subsubsection{Check refutation.}
If $\mainloop$ does not report safety, it reports a subsystem that hints to a \emph{potential} counterexample.
Formally, this subsystem is a subMDP of states that were `visited' during the invocation of $\strengthen$.%
\begin{definition}[subMDP]
\label{def:submdp}
	Let $\mdp = \mdps$ be an MDP and let $S' \subseteq S$ with $\sinit \in S'$.
	We call $\mdp_{S'} = \mdpsprime$ the \emph{subMDP induced by $\mdp$ and $S'$}, where for all $s,s' \in S'$ and all $a \in \act$, we have $ \transmatrix' (s, a, s') = \transmatrix (s, a, s') $. 
    \hfill $\triangle$
\end{definition}%
A subMDP $\mdp_{S'}$ may be substochastic where missing probability mass never reaches a bad state.
Def.~\ref{def:mdps} is thus relaxed: For all states $s \in S'$ we require that $\sum_{s'\in S'} \transmatrix(s,a,s') \leq 1$.%
If the subsystem is unsafe, we can conclude that the original MDP $\mdp$ is also safe.
\begin{lemma}
	If $\mdp'$ is a subMDP of $\mdp$ and $\mdp'$ is unsafe, then $\mdp$ is also unsafe. 
\end{lemma}%
The role of $\checkrefutation$ is to establish whether the subsystem is indeed a true counterexample or a spurious one.
Formally, $\checkrefutation$ should ensure:
	   \begin{align*}
	      \htriple{\true}
          {\textit{res} \gets \checkrefutation\left( \subsystem \right)}
          {\textit{res} = \true ~{}\Leftrightarrow{}~ \smdp{\subsystem} \text{ unsafe}}.  
	   \end{align*}
Again, $\outerloop$ is backward compatible in the sense that a single fixed heuristic is always sufficient when reasoning about reachability ($\thresh = 0$).%
\begin{recovery}
	\label{recovery:5}
For qualitative reachability~($\lambda = 0$) and the heuristic $\heuristic_0$ from \textnormal{Recovery Statement~\ref{recovery:4}}, \outerloop 
invokes its core $\mainloop$ exactly \emph{once}.
\end{recovery}%
This statement is true, as $\mainloop$ returns either $\safeReturn$ or a subsystem containing a path from the initial state to a bad state. 
In the latter case, $\checkrefutation$ detects that the subsystem is indeed a counterexample which cannot be spurious in the qualitative setting.

We remark that the procedure $\checkrefutation$ invoked in l.~\ref{alg:outer:refute}
is a classical fallback;
it runs an (alternative) model checking algorithm, e.g., solving the set of Bellman equations, for the subsystem.
In the worst case, i.e., for $S' = S$, we thus solve exactly our problem statement.
Empirically (Tab.~\ref{tab:primaryresults_1}) we observe that for reasonable oracles the procedure $\checkrefutation$ is invoked on significantly smaller subMDPs. However, in the worst case the subMDP must include 
\emph{all} paths of the original MDP, and then thus coincides.

\subsubsection{Refine oracle.}
Whenever we have neither proven the MDP safe nor unsafe,
we refine the oracle to prevent generating the same subsystem in the next invocation of $\mainloop$.
To ensure termination, oracles should only be refined finitely often.
That is, we need some progress measure. 
The set $\touched$ overapproximates all counterexamples encountered in some invocation of $\mainloop$ and we propose to use its size as the progress measure.
While there are several possibilities to update $\touched$ through the user-defined procedure $\interface{\enlarge}$~(l.~\ref{alg:outermost:touched}), every implementation should hence satisfy
$\htriple{\true}{\touched' \gets \interface{\enlarge(\touched,\blank\,)}}{|\touched'| > |\touched|}$.
Consequently, after finitely many iterations, the oracle is refined with respect to all states. 
In this case, we may as well rely on solving the characteristic LP problem: 
\begin{lemma}
\label{lem:finallyperfect}
The algorithm \outerloop in Alg.~\ref{alg:outermost_loop} is sound and complete
if  $\refineoracle(\oracle,S)$ returns a perfect oracle $\oracle$ (with $S$ is the set of all states).
\end{lemma}
Weaker assumptions on $\refineoracle$ are possible, but are beyond the scope of this paper. 
Moreover, the above lemma does not rely on the abstract concept that heuristic $\heuristic$ provides suitable probabilities after finitely many refinements.\footnote{One could of course now also create a heuristic that is trivial for a perfect oracle and invoke \mainloop{} with the heuristic for the perfect oracle, but there really is no benefit in doing so.}


\newcommand{\pvarfont}[1]{\mathsf{#1}}

\newcommand{\var}{\pvarfont{vars}}
\newcommand{\varc}{\pvarfont{c}}
\newcommand{\varf}{\pvarfont{f}}
\newcommand{\valc}{\pvarfont{val\_ c}}
\newcommand{\valf}{\pvarfont{val\_ f}}
\newcommand{\Frame}{\pvarfont{Frame}}
\newcommand{\Frameapp}[1]{\pvarfont{Frame} \left( #1 \right)}
\newcommand{\Goal}{\pvarfont{Bad}}
\newcommand{\Zero}{\pvarfont{Zero}}
\newcommand{\Goalapp}[1]{\pvarfont{Bad} \left( #1 \right)}
\newcommand{\smtbm}{\pvarfont{Phi}}
\newcommand{\smtbmapp}[1]{\pvarfont{Phi} \left( #1 \right) }
\newcommand{\chosencmd}{\pvarfont{ChosenCommand}}
\newcommand{\States}{\pvarfont{States}}

\newcommand{\niceforall}{\forall \,}

\section{Practical \boldpricthree}
\label{sec:practical}

So far, we gave a conceptual view on \prict, but now take a more practical stance. 
We detail important features of effective implementations of \pricthree (based on our empirical evaluation).
We first describe an implementation without generalization,
and then provide a prototypical extension that allows for three variants of generalization. 

\lstset{language=prism}
\newsavebox{\prismexample}
\begin{lrbox}{\prismexample}
\begin{lstlisting}[numbers=none]
module ex
  c : [0..20] init 0;   f : [0..1] init 0;
  [] c<20 -> 0.1:(f'=1) + 0.9:(c'=c+1); // cmd 1
  [] c<10 -> 0.2:(f'=1) + 0.8:(c'=c+2);
endmodule
\end{lstlisting}
\end{lrbox}
\begin{figure}[t]
\centering
\subfloat[\prism{} code snippet]{
\raisebox{0.6cm}{\usebox{\prismexample}}
}
\subfloat[Part of the corresponding MDP]{
\begin{tikzpicture}
	\node[rectangle,draw] (c2f0) {\scriptsize{$c{=}2{,}f{=}0$}};
	\node[rectangle,below=0.4cm of c2f0,draw] (c2f1) {\scriptsize{$c{=}2{,}f{=}1$}};
	\node[rectangle,left=0.4cm of c2f1,draw] (c3f0) {\scriptsize{$c{=}3{,}f{=}0$}};
	\node[rectangle,right=0.4cm of c2f1,draw] (c4f0) {\scriptsize{$c{=}4{,}f{=}0$}};
	\node[circle,fill, inner sep=1pt,left=0.7 cm of c2f0] (a1) {};
	\node[circle,fill, inner sep=1pt,right=0.7 cm of c2f0] (a2) {};
	
	\draw[-] (c2f0) edge node[above]{\tiny{cmd1}} (a1);
	\draw[-] (c2f0) edge node[above]{\tiny{cmd2}} (a2);
	\draw[->] (a1) edge node[left]{\tiny{$0.1$}} (c2f1);
	
	\draw[->] (a1) edge node[left]{\tiny{$0.9$}} (c3f0);
	\draw[->] (a2) edge node[right]{\tiny{$0.2$}} (c2f1);
	\draw[->] (a2) edge node[right]{\tiny{$0.8$}} (c4f0);
\end{tikzpicture}
\label{fig:prismex:mdp}
}
\caption{Illustrative \prism-style probabilistic guarded command language example}
\label{fig:prismex}
\end{figure}
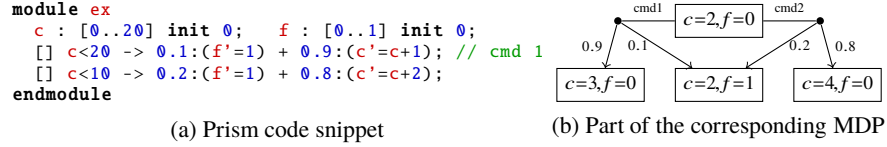

\subsection{A Concrete \pricthree Instance without Generalization}
\label{sec:ic3instantiated}

\paragraph{Input.}
We describe MDPs using the \prism{} guarded command language\footnote{
Preprocessing ensures a single thread (module) and no deadlocks.}, exemplified in Fig.~\ref{fig:prismex}. 
States are described by valuations to $m$ (integer-valued) program variables $\vars$, 
and outgoing actions are described by commands of the form \begin{align*}\footnotesize{\text{\texttt{[] guard -> prob1 : update1 \& ... \& probk : updatek }}}\end{align*}
If a state satisfies \texttt{guard}, then the corresponding action with $k$ branches exists; probabilities are given by \texttt{probi}, the successor states are described by \texttt{updatei}, see Fig.~\ref{fig:prismex:mdp}.

\paragraph{Encoding.}
We encode frames as logical formulae. Updating frames then corresponds to adding conjuncts, and checking for relative inductivity is a satisfiability call.
Our encoding is as follows:
States are assignments to the program variables, i.e., $\States = \Ints^m$.
We use various uninterpreted functions, to whom we give semantics using  appropriate constraints.
Frames\footnote{
In each operation, we only consider a single frame.} are represented by uninterpreted functions $\fun{\Frame}{\States}{\Reals}$
satisfying $\Frameapp{s} = d$ implies $F[s]\geq d$.
Likewise, the Bellman operator is an uninterpreted function  $\fun{\smtbm}{\States}{\Reals}$ such that $\smtbmapp{s} = d$ implies $\bm{F}[ s ] \geq d$.
Finally, we use $\fun{\Goal}{\States}{\Bools}$ with  
	           $\Goalapp{s}$ iff $s \in \bad$.

%

Among the appropriate constraints, we ensure that variables are within their range, bound the values for the frames, and enforce $\smtbmapp{s}=1$ for $s \in \bad$.
We encode the guarded commands as exemplified by this encoding of the first command in Fig.~\ref{fig:prismex}:
\begin{align*}
& \niceforall s \in \States \colon \neg \Goalapp{s} \wedge 
		             s[c] < 20   \\&\quad 
		             \Longrightarrow
		          \smtbmapp{s} =  0.1 \cdot \Frameapp{(s[c],1)} + 0.9 \cdot \Frameapp{(s[c] + 1,s[f])}.
\end{align*}
In our implementation, we optimize the encoding.
We avoid the uninterpreted functions by applying an adapted Ackerman reduction.
We avoid universal quantifiers, by first observing that we always ask whether a single state is not inductive, and then unfolding the guarded commands in the constraints  that describe a frame.
That encoding grows linear in the size of the maximal out-degree of the MDP, and is in the quantifier-free fragment of linear arithmetic (QFLRIA).

%

\paragraph{Heuristic.}
We select probabilities $\delta_i$ by solving the following optimization problem, with variables $x_i$, $\textsl{range}(x_i)\in [0,1]$, for states $s_i \in \succs(s,a)$ and oracle $\oracle$\footnote{If $\max \oracle(s_j) = 0$, we assume $\forall j. \oracle(s_j) = 0.5$. If $\delta = 0$, we omit rescaling to allow $\sum x_j =0$.}.
\begin{align*}
       &\text{minimize} \sum^k_{\substack{i \\  s_i \not \in \bad}} \abs{\frac{x_i}{\sum_{j=1}^k x_j} - \frac{\oracle \left( s_i \right)}{\sum_{j=1}^n \oracle \left( s_j \right)}} \notag \; \text{s.t.}\;
       %
       \delta = \sum_{i=1}^k P(s, a , s_i) \cdot 
       \begin{cases}
       1, & \text{if}~s_i \in \bad, \\
       x_i, & \text{else.}
       \end{cases} \notag 
       %
    %
\end{align*}%
The constraint ensures that, if the values $x_i$ correspond to the actual reachability probabilities from $s_i$, then the reachability from state $s$ is exactly $\delta$. 
A constraint stating that $\delta \geq \hdots$ would also be sound, but we choose equality as it preserves room between the actual probability and the threshold we want to show.
Finally, the objective function aims to preserve the ratio between the suggested probabilities.

\paragraph{Repushing and breaking cycles.}
\emph{Repushing}
~\cite{DBLP:conf/fmcad/EenMB11}
is an essential ingredient of both standard \icthree and \pricthree.
Intuitively, we avoid opening new frames and spawning obligations that can be deduced from current information. 
Since repushing generates further obligations in the current frame, its implementation requires that the detection of Zeno-behavior 
has to be moved from $\mainloop$ into the \strengthen{} procedure. 
%
Therefore, we track the histories of the obligations in the queue.
Furthermore, once we detect a cycle we first try to adapt the heuristic $\heuristic$ locally to overcome this cyclic behavior instead of immediately giving up. This local adaption reduces the number of $\mainloop$ invocations.
    
\paragraph{Extended queue.}
In contrast to standard \icthree, the obligation queue might contain entries that vary only in their $\delta$ entry. 
In particular, if the MDP is not a tree, it may occur that the queue contains both  $(i, s, \delta)$ and $(i, s, \delta')$ with $\delta > \delta'$.
Then, $(i, s, \delta')$ can be safely pruned from the queue.
Similarly, after handling $(i, s, \delta)$, if some fresh obligation $(i,  s, \delta'' > \delta)$ is pushed to the queue, it can be substituted with $(i, s, \delta)$.
To efficiently operationalize these observations, we keep an additional mapping which remains intact  over multiple invocations of \strengthen.
We furthermore employed some optimizations for $\qtouched$ aiming to track potential counterexamples better.
%
After refining the heuristic, one may want to reuse frames or the obligation queue, but empirically this leads to performance degradation as the values in the frames are inconsistent with behavior suggested by the heuristic.

\subsection{Concrete \pricthree with Generalization}
\label{sec:generalization}
So far, frames are updated by changing single entries whenever we resolve obligations $(i, s, \delta)$, i.e., we add conjunctions of the form $F_i[s] \leq \delta$. 
Equivalently, we may add a constraint $\forall s' \in S: F_i[s'] \leq p_{\{s\}}(s')$ with $p_{\{s\}}(s) = \delta$ and $p_{\{s\}} = 1$ for all $s' \neq s$.

Generalization in \icthree aims to update a set $\genset$ (including $s$) of states in a frame rather than a single one without invalidating relative inductivity.
In our setting, we thus consider a function $p_{\genset}\colon \genset \to [0,1]$ with $p_{\genset}(s) \leq \delta$ that assigns (possibly different) probabilities to all states in $\genset$.
Updating a frame then amounts to adding the constraint
\[\niceforall s\in \States\colon s \in \genset \Longrightarrow \Frameapp{s} \leq \genfunc(s).\]
Standard \icthree{} generalizes by iteratively ``dropping'' a variable, say $v$. The set $\genset$ then consists of all states that do not differ from the fixed state $s$ except for the value of $v$.\footnote{Formally, $\genset = \{ s' \mid  \text{for all } v' \in \vars \setminus \{ v \}: s'(v') = s(v')\}$.}
We take the same approach by iteratively dropping program variables. Hence, $p_{\genset}$ effectively becomes a mapping from the value $s[v]$ to a probability.
We experimented with four types of functions $\genfunc$ that we describe for Markov chains.
The ideas are briefly outlined below; details are beyond the scope of this paper. 

\paragraph{Constant $\genfunc$.}
Setting all $s \in \genset$ to $\delta$ is straightforward but empirically not helpful.

\paragraph{Linear interpolation.}
We use a linear function $\genfunc$ that interpolates two points.
The first point $(s[v], \delta)$ is obtained from the obligation $(i,s,\delta)$.
 For a second point, consider the following: 
Let $\textnormal{Com}$ be the unique\footnote{Recall that we have a Markov chain consisting of a single module} command active at state $s$.
	             Among all states in $\genset$ that are enabled in the guard of $\textnormal{Com}$,
	             we take the state $s'$ in which $s'[v]$ is maximal\footnote{This implicitly assumes that $v$ is increased. Adaptions are possible.}. 
	  The second point for interpolation is then $(s'[v], \bm{F_{i-1}}[s'])$.
	   If the relative inductivity fails for $\genfunc$ we do not generalize with $\genfunc$, but may attempt to find other functions.

\paragraph{Polynomial interpolation.}
Rather than linearly interpolating between two points, we may interpolate using more than two points. In order to properly fit these points, we can use a higher-degree polynomial. We select these points using counterexamples to generalization (CTGs):
We start as above with linear interpolation.
However, if $\genfunc$ is not relative inductive, the
SMT solver yields a model with state $s'' \in \genset$ and probability $\delta''$, with $s''$ violating relative inductivity, i.e., $\bm{F_{i-1}}[s''] > \delta''$. We call $(s'', \bm{F_{i-1}}[s''])$ a CTG, and $(s''[v],\bm{F_{i-1}}[s'']))$ is then a further interpolation point, and we repeat.
	    
   Technically, when generalizing using nonlinear constraints, we use real-valued arithmetic with a branch-and-bound-style approach to ensure integer values.

\paragraph{Hybrid interpolation.}
In polynomial interpolation, we generate high-degree polynomials and add them to the encoding of the frame. In subsequent invocations, reasoning efficiency is drastically harmed by these high-degree polynomials.
Instead, we soundly approximate $\genfunc$ by a piecewise linear function, and use these constraints in the frame.


%

\section{Experiments}
\label{sec:experiments}
We 
assess 
how \pricthree{} may contribute to the state of the art in probabilistic model checking.
We do some early empirical evaluation showing that \pricthree{} is feasible. We see ample room for further improvements of the prototype.

\paragraph{Implementation.}
We implemented a prototype\footnote{The prototype is available open-source from \url{https://github.com/moves-rwth/PrIC3}.} of \pricthree{} based on Sect.~\ref{sec:ic3instantiated} in Python.
The input is represented using efficient data structures provided by the model checker \storm{}.
We use an incremental instance of Z3~\cite{DBLP:conf/tacas/MouraB08} for each frame, as suggested in~\cite{DBLP:conf/fmcad/EenMB11}. A solver for each frame is important to reduce the time spent on pushing the large frame-encodings.
The optimization problem in the heuristic is also solved using Z3.
All previously discussed generalizations (none, linear, polynomial, hybrid) are supported. 

\paragraph{Oracle and refinement.}
We support the (pre)computation of four different types of oracles 
for the \interface{initialization} step in Alg.~\ref{alg:outermost_loop}:
(1)~A perfect oracle solving \emph{exactly} the Bellman equations. Such an oracle is unrealistic, but interesting from a conceptual point.
(2)~Relative frequencies by recording all visited states during simulation. This idea is a na\"ive simplification of Q-learning.
(3)~Model checking with decision diagrams (DDs) and few value iterations. Often, a DD representation of a model can be computed fast, and the challenge is in executing sufficient value iterations. We investigate whether doing few value iterations yields a valuable oracle (and covers states close to bad states).
(4)~Solving a (pessimistic) LP from BFS partial exploration. States that are not expanded are assumed bad. Roughly, this yields oracles covering states close to the initial states.

To implement \interface{\refineoracle} (cf.\ Alg.~\ref{alg:outermost_loop}, l.~\ref{alg:outermost:refine}), we create an LP for the subMDP induced by the touched states. For states whose successors are not in the touched states, we add a transition to $B$ labeled with the oracle value as probability. 
The solution of the resulting LP updates the entries corresponding to the touched states.

For \interface{\enlarge} (cf.\ Alg.~\ref{alg:outermost_loop}, l.~\ref{alg:outermost:touched}), we take the union of the subsystem and the touched states. 
If this does not change the set of touched states, we also add its successors.


\paragraph{Setup.}
We evaluate the run time and memory consumption of our prototype of \pricthree.
We choose a combination of models from the literature (BRP~\cite{DBLP:conf/papm/DArgenioJJL01}, ZeroConf~\cite{DBLP:journals/rfc/rfc3927}) and some structurally straightforward variants of grids (chain, double chain; see App.~\ref{app:experimentdetails}).
 Since our prototype lacks the sophisticated preprocessing applied by many state-of-the-art model checkers, it is more sensitive to the precise encoding of a model, e.g., the number of commands.
 To account for this, we generated new encodings for all models.
%
All experiments were conducted on an single core of an
Intel® Xeon® Platinum 8160 processor. 
We use a 15 minute time-limit and report TO otherwise. Memory is limited to 8GB; we report MO if it is exceeded.
Apart from the oracle, all parameters of our prototype remain fixed over all experiments.
To give an impression of the run times, we compare our prototype with both the explicit (Storm$_\text{sparse}$) and DD-based (Storm$_\text{dd}$) engine of the model checker \storm~1.4, which compared  favourably in QComp~\cite{DBLP:conf/tacas/HahnHHKKKPQRS19}.

\paragraph{Results.}
In Tab.~\ref{tab:primaryresults_1}, we present the run times for various invocations of our prototype and Oracle~4\footnote{We explore $\min \{|S|, 5000\}$ states using BFS and \storm.}.
In particular, we give the model name and the number of (non-trivial) states in the particular instance, and the (estimated) actual probability to reach  $B$.
For each model, we consider multiple thresholds $\thresh$.
\sebastianmargin{please check below.}
The next 8 columns report on the four variants of \pricthree{} with varying generalization schemes. Besides the scheme with the run times, we report for each scheme the number of states of the largest (last) subsystem that $\checkrefutation$ in Alg.~\ref{alg:outermost_loop}, l.~\ref{alg:outer:refute} was invoked upon (column $|sub|$).
The last two columns report on the run times for \storm{} that we provide for comparison.
In each row, we mark with {\color{red!50!blue}purple}
 MDPs that are unsafe, i.e., \pricthree{} refutes these MDPs for the given threshold $\thresh$. We \textbf{highlight} the best configurations of \pricthree.
 


\begin{table}[t]
\caption{Empirical results. Run times are in seconds; time out = 15 minutes.}
\centering

{\scriptsize	
\begin{tabular}{c r r  r || r r | r r | r r | r r | r  r}

 &\phantom{a} $|S|$ &\phantom{a} $\prevtmax{\mc}{\sinit}{\bad}$  &\phantom{a}  $\lambda$ &\phantom{a} w/o &\phantom{a} $|\text{\textit{sub}}|$ &\phantom{a}  lin &\phantom{a} $|\text{\textit{sub}}|$& \phantom{a}  pol &\phantom{a} $|\text{\textit{sub}}|$ &\phantom{a}  hyb &\phantom{a} $|\text{\textit{sub}}|$ &\phantom{a} Storm$_\text{sparse}$ &\phantom{a}  Storm$_\text{dd}$ \\[0.5ex]
 \hline \hline
 \multirow{3}{*}{\rotatebox{90}{BRP}} &  \multirow{3}{*}{$10^3$}  & \multirow{3}{*}{$0.035$} &  $0.1$ & TO & -- & TO & -- & TO & -- &  TO & --  & ${<}0.1$ &  $0.12$  \\  
&&& {\color{red!50!blue}$0.01$} & $\mathbf{51.3}$ & $324$ & $125.8$ & $324$  & TO & -- & MO & -- & ${<}0.1$ &  $0.18$ \\
& & & {\color{red!50!blue}$0.005$} & $\mathbf{10.9}$ & $188$ & $38.3$ & $188$ & TO & -- & MO & -- & ${<}0.1$ &  $0.1$\\
  \hline

\multirow{7}{*}{\rotatebox{90}{ZeroConf}}  &  \multirow{3}{*}{$10 ^{4}$}  & \multirow{3}{*}{$0.5$} &  $0.9$  & TO & -- & TO & -- &  $0.4$ & $0$ & $\mathbf{0.1}$ & $0$ & ${<}0.1$ & $296.8$  \\ 
& & & $0.52$ & TO & -- & TO & -- & $0.2$ & $0$ & $\mathbf{0.2}$ & $0$&  ${<}0.1$ & $282.6$ \\ 
& & & {\color{red!50!blue}$0.45$} & $\mathbf{{<}0.1}$ & $1$ & $\mathbf{{<}0.1}$ & $1$ & $\mathbf{{<}0.1}$ & $1$ & $\mathbf{{<}0.1}$ & $1$ & ${<}0.1$ & $300.2$   \\\cline{3-14} 
&  \multirow{4}{*}{$10^9$} & \multirow{4}{*}{${\sim}0.55$} & $0.9$ & TO & -- & TO & -- &  $\mathbf{3.7}$ & $0$ & MO & --  &  MO &  TO \\ 
&&&$0.75$ & TO & -- & TO & -- &  $\mathbf{3.4}$ & $0$ & MO & --  & MO & TO  \\
 & & & $0.52$ & TO & -- & TO & -- & TO & -- & TO & -- & MO & TO \\
& & & {\color{red!50!blue}$0.45$} & $\mathbf{{<}0.1}$ & $1$ & $\mathbf{{<}0.1}$ & $1$ & $\mathbf{{<}0.1}$ &$1$ & $\mathbf{{<}0.1}$ & $1$ & MO & TO   \\\hline
 \multirow{9}{*}{\rotatebox{90}{Chain}} &  \multirow{4}{*}{$ 10 ^3$}  & \multirow{4}{*}{$0.394$} &  $0.9$ & $18.8$ & $0$ & $60.2$ & $0$ &$1.2$ & $0$ & $\mathbf{0.3}$ & $0$ & ${<}0.1$ &  ${<}0.1$  \\ 
& & & $0.4$ & $20.1$ & $0$ & $55.4$ & $0$ & $\mathbf{0.9}$ & $0$ & TO  & --  & ${<}0.1$ &  ${<}0.1$  \\
& & & {\color{red!50!blue}$0.35$} & $\mathbf{91.8}$ & $431$ & $119.5$ & $431$ & TO & -- & TO & -- & ${<}0.1$ &  ${<}0.1$ \\
& & & {\color{red!50!blue}$0.3$} & $\mathbf{46.1}$ & $357$ & $64.0$ & $357$ & TO & -- & TO & -- & ${<}0.1$ &  ${<}0.1$ \\\cline{3-14} 
&  \multirow{3}{*}{$10 ^4$} & \multirow{3}{*}{$0.394$} & $0.9$ & TO & -- & TO & -- &  $1.6$ & $0$ & $\mathbf{0.3}$ & $0$ & ${<}0.1$ & $4.5$ \\
& & & $0.4$ & TO & -- & TO & -- & $\mathbf{1.4}$ & $0$ & TO & --  & ${<}0.1$ & $4.9$  \\
& & & {\color{red!50!blue}$0.3$} & TO & -- & TO & -- & TO & -- & TO & --   & ${<}0.1$ & $4.9$  \\\cline{3-14}
&  \multirow{2}{*}{$10 ^{12}$}  & \multirow{2}{*}{$0.394$} &  $0.9$ & TO & -- & TO & -- &  $\mathbf{6.4}$ & $0$ & MO & -- & MO & TO   \\
& & & $0.4$ & TO & -- & TO & -- & $\mathbf{6.0}$ & $0$ & MO & --  & MO & TO  \\
\hline

\multirow{8}{*}{\rotatebox{90}{Double Chain}} &  \multirow{4}{*}{$ 10 ^3$}  & \multirow{4}{*}{$0.215$} &  $0.9$ & $528.1$ & $0$ & $828.8$ & $0$ & $203.3$ & $0$ & $\mathbf{0.6}$ & $0$ & ${<}0.1$ & ${<}0.1$   \\  
&&&$0.3$ & $588.4$ & $0$ & TO & -- & $138.3$ & $0$  & $\mathbf{0.5}$ & $0$ & ${<}0.1$ & ${<}0.1$ \\
& & & $0.216$ &  $\mathbf{597.4}$ & $0$ & TO & --  & $765.8$ & $0$ & MO & -- & ${<}0.1$ & ${<}0.1$   \\   
& & & {\color{red!50!blue}$0.15$} & TO & -- & TO & -- & TO & -- & TO & -- & ${<}0.1$ & ${<}0.1$   \\\cline{3-14}    
%
%
& \multirow{2}{*}{$10 ^{4}$}  & \multirow{2}{*}{$0.22$}& $0.3$ & TO & --  & TO & -- & $17.5$ & $0$ & $\mathbf{0.5}$ & $0$ & $0.2$ & $2.6$ \\
& & & {$0.24$} & TO & -- & TO & --  &$\mathbf{16.8}$ & $0$ & MO & --  & $0.2$ & $2.7$ \\\cline{3-14} 
&  \multirow{2}{*}{$10^7$} & \multirow{2}{*}{$2.6\eexp{-4}$} & $4\eexp{-3}$ &  TO & -- & TO & -- & TO & --  & MO & -- &  TO & TO   \\
& & & $2.7\eexp{-4}$  &  TO & -- & TO & -- & $\mathbf{281.2}$ & $0$ &  MO & -- & TO & TO  \\ 

\end{tabular}
}

\label{tab:primaryresults_1}
\end{table}

\paragraph{Discussion.}
Our experiments give a mixed picture on the performance of our implementation of \pricthree{}.
On the one hand, Storm significantly outperforms \pricthree{} on most models.
On the other hand, \pricthree{} is capable of reasoning about huge, yet simple, models with up to $10^{12}$ states that Storm is unable to analyze within the time and memory limits.
There is more empirical evidence that \pricthree{} may complement the state-of-the-art:

First, \emph{the size of thresholds matters}.
Our benchmarks show that---at least without generalization---more ``wiggle room'' between the precise maximal reachability probability and the threshold generally leads to a better performance. 
\pricthree{} may thus prove bounds for large models where a precise quantitative reachability analysis is out of scope.

Second, \emph{\pricthree{} enjoys the benefits of bounded model checking}. In some cases, e.g., ZeroConf for $\lambda = 0.45$, \pricthree{} refutes very fast as it does not need to \mbox{build the whole model.}

Third, if \pricthree{} proves the safety of the system, it  does so without relying on checking large subsystems in the $\checkrefutation$ step.

Fourth, \emph{generalization is crucial}. Without generalization, \pricthree{} is unable to prove safety for any of the considered models with more than $10^3$ states.
With generalization, however, it can prove safety for very large systems and thresholds close to the exact reachability probability. For example, it proved safety of the Chain benchmark with $10^{12}$ states for a threshold of $0.4$ which differs from the exact reachability probability by $0.006$.

Fifth, \emph{there is no best generalization}. There is no clear winner out of the considered generalization approaches. Linear generalization always performs worse than the other ones. In fact, it performs worse than no generalization at all. The hybrid approach, however, occasionally has the edge over the polynomial approach. 
This indicates that more research is required to find suitable generalizations.

In Appendix~\ref{app:experimentdetails}, we also compare the additional three types of oracles (1--3). 
We observed that only few oracle refinements are needed to prove \emph{safety}; for small models at most one refinement was sufficient.
However, this does not hold if the given MDP is unsafe. 
DoubleChain with $\thresh = 0.15$, for example, and Oracle 2 requires 25 refinements.


%
%

\section{Conclusion}
We have presented \prict---the first truly probabilistic, yet conservative, extension of \icthree to quantitative reachability in MDPs. 
Our theoretical development is accompanied by a prototypical implementation and experiments.
We believe there is ample space for improvements including an in-depth investigation of suitable oracles and generalizations.


%
\bibliographystyle{splncs04}
\bibliography{literature}


\clearpage

\appendix
\section{Further details for the experiments}
\label{app:experimentdetails}

\begin{figure}[h!]
\begin{lstlisting}[numbers=none]
module chain
c : [0..N] init 0;   f : [0..1] init 0;
[] c<N-> p:(c'=c+1) + (1-p):(f'=1); 
endmodule
label bad = f'=1
\end{lstlisting}
\begin{lstlisting}[numbers=none]
module double_chain
c : [0..N] init 0;   f : [0..1] init 0;   g : [0..1] init 0;
[] c<N & g=0 -> p1:(c'=c+1) + p2:(g'=1) + p3 : (f' =1); 
[] c<N & g=1 -> q:(c'=c+1) + (1-q): (f'=1);
endmodule
label bad = f'=1
\end{lstlisting}
\caption{Prism code for examples Chain and Double Chain}
\label{fig:chains}
\end{figure}
In Fig.~\ref{fig:chains}, we depict the precise encoding of the two variants of chains.

For Table~\ref{tab:primaryresults_2} (alternative oracles), we use only the smallest models  from Table~\ref{tab:primaryresults_1} as we focus on the effect of the oracles. The model name and the threshold $\lambda$ thus identify the instance; $\oracle$ refers to the oracle type.\footnote{We use 100 value iteration steps for Oracle~3.} We then list the generalization approach. For each configuration, we report the number of iterations of $\mainloop$ and the run time in seconds.

Table~\ref{tab:primaryresults_1_appendix} provides additional empirical results for the benchmarks from Table~\ref{tab:primaryresults_1}: The number of calls to $\checkrefutation$ and the percentage of the total runtime required by $\checkrefutation$.
\begin{table}[h!]
\caption{Comparison of the number iterations of \pricthree{}'s outermost loop and its run times (in seconds) for three different oracles $\oracle$.}
\centering

\newcommand{\mcoq}{$m_{100}$}
\newcommand{\mcoh}{$m_{250}$}
{\scriptsize	
	%
	%
	%
	%
	%
	\begin{tabular}{c  c  c ||  r r |  r r | r r | r r }
		
		\phantom{aaa}& $\thresh$ &  $\oracle$   & \multicolumn{2}{c|}{w/o}& \multicolumn{2}{c|}{lin}& \multicolumn{2}{c|}{pol}& \multicolumn{2}{c}{hyb}  \\[0.5ex]
		\hline \hline
		%
		%
		\multirow{9}{*}{\rotatebox{90}{ZeroConf}}  &   \multirow{3}{*}{$0.75$} &  1 &  - & TO &  - & TO & - & TO &  - & MO   \\
		&&2 &  - & TO &  - & TO &  $0$ & $0.3$ &  $\mathbf{0}$ & $\mathbf{0.1}$ \\
		&&3 &  - & TO &  - & TO &  $\mathbf{1}$ & $\mathbf{0.2}$ &  $\mathbf{1}$ & $\mathbf{0.2}$ \\\cline{2-11}
		&\multirow{3}{*}{$0.52$}&  1 &  - & TO &  - & TO &  - & TO &  - & MO   \\
		&&2 &  - & TO &  - & TO &  $0$ & $0.3$ &  $0$ & $0.1$ \\
		&&3 &  - & TO &  - & TO &  $\mathbf{1}$ & $\mathbf{0.2}$ &  $\mathbf{1}$ & $\mathbf{0.2}$ \\\cline{2-11}
		&\multirow{3}{*}{{\color{red!50!blue}$0.45$}}&  1 &  - & TO &  - & TO &  - & TO &  - & MO   \\
		&&2 &  $\mathbf{1}$ & $\mathbf{<0.1}$ &  $\mathbf{1}$ & $\mathbf{<0.1}$ &  $\mathbf{1}$ & $\mathbf{<0.1}$ &  $\mathbf{1}$& $\mathbf{<0.1}$ \\
		&&3 &  $\mathbf{1}$ & $\mathbf{<0.1}$ &  $\mathbf{1}$ & $\mathbf{<0.1}$ &  $\mathbf{1}$ & $\mathbf{<0.1}$ &  $\mathbf{1}$ & $\mathbf{<0.1}$ \\\hline\hline

		\multirow{9}{*}{\rotatebox{90}{Chain}}  &   \multirow{3}{*}{$0.9$} &  1 & $0$ & $20.72$ &  $0$ & $65.7$ &  $\mathbf{0}$ & $\mathbf{0.1}$ &  $\mathbf{0}$ & $\mathbf{0.1}$   \\
		&&2 &  $0$ & $19.1$ &  $0$ & $64.9$ &  $0$ & $1.1$ &  $\mathbf{0}$ & $\mathbf{0.2}$ \\
		&&3 &  $0$ & $22.3$ &  $0$ & $65.7$ &  $0$ & $1.1$ &  $\mathbf{0}$ & $\mathbf{0.2}$ \\\cline{2-11}
		&\multirow{3}{*}{$0.4$}&  1 &  $0$ & $23.4$ &  $0$ & $61.9$ &  $\mathbf{0}$ & $\mathbf{22.1}$ &  - & MO   \\
		&&2 &  $0$ & $21.2$ &  $0$ & $61.9$ &  $\mathbf{0}$ & $\mathbf{1.0}$ &  - & MO\\
		&&3 &  $0$ & $23.4$ &  $0$ & $65.1$ &  $\mathbf{0}$ & $\mathbf{1.0}$ &  - & MO \\\cline{2-11}
		&\multirow{3}{*}{{\color{red!50!blue}$0.35$}}&  1 &  $\mathbf{1}$ & $\mathbf{77.8}$ &  $1$ & $99.5$ &  - & TO &  - & MO   \\
		&&2 &  $\mathbf{1}$ & $\mathbf{79.1}$ &  $1$ & $107.3$ &  - & TO &  - & MO \\
		&&3 &  $\mathbf{1}$ & $\mathbf{90.1}$ &  $1$ & $123.2$ &  - & TO &  - & MO \\\hline\hline

		\multirow{9}{*}{\rotatebox{90}{Double Chain}}  &   \multirow{3}{*}{$0.3$} &  1 &  - & TO &  - & TO &  - & TO &  $\mathbf{1}$ & $\mathbf{9.9}$   \\
		&&2 &  - & TO &  - & TO &  $0$ & $302.7$ &  $\mathbf{0}$ & $\mathbf{0.4}$ \\
		&&3 &  $0$ & $192.5$ &  $0$& $223.4$ &  $0$ & $28.4$ &  $\mathbf{0}$ & $\mathbf{0.5}$ \\\cline{2-11}
		&\multirow{3}{*}{$0.216$}&  1 &  - & TO &  - & TO &  - & TO &  - & MO   \\
		&&2 &  - & TO &  - & TO &  $\mathbf{0}$ & $\mathbf{245.5}$ &  - & MO \\
		&&3 &  - & TO &  - & TO &  - & TO &  - & MO \\\cline{2-11}
		&\multirow{3}{*}{{\color{red!50!blue}$0.15$}}&  1 &  - & TO &  - & TO &  - & TO &  - & MO   \\
		&&2 &  $25$ & $\mathbf{373.5}$ &  $\mathbf{18}$ & $701.3$ &  - & TO &  - & MO \\
		&&3 &  - &  TO &  - & TO &  - & TO &  - & MO \\\hline

	\end{tabular}
}
\label{tab:primaryresults_2}
\end{table}

\begin{table}[t]
	\caption{Additional empirical results.  $\#cr$ is the number of calls to $\checkrefutation$.
	$\text{crTime}\, (\%)$ is the percentage of the total runtime required by $\checkrefutation$.}
	\centering
	
{\scriptsize	
\begin{tabular}{c r r  r || r r | r r | r r | r r }

&&&&\multicolumn{2}{c}{w/o}&\multicolumn{2}{c}{lin}&\multicolumn{2}{c}{pol}&\multicolumn{2}{c}{hyb} \\ [0.5ex]
 &\phantom{a} $|S|$ &\phantom{a} $\prevtmax{\mc}{\sinit}{\bad}$  &\phantom{a}  $\lambda$ &\phantom{a} \#cr &\phantom{a} $\text{crTime}\, (\%)$ &\phantom{a}  \#cr &\phantom{a} $\text{crTime}\, (\%)$ & \phantom{a}  \#cr &\phantom{a} $\text{crTime}\, (\%)$ &\phantom{a}  \#cr &\phantom{a} $\text{crTime}\, (\%)$  \\[0.5ex]
 \hline \hline
 \multirow{3}{*}{\rotatebox{90}{BRP}} &  \multirow{3}{*}{$10^3$}  & \multirow{3}{*}{$0.035$} &  $0.1$ & -- & -- & -- & -- & -- & -- &  -- & --   \\  
&&& {\color{red!50!blue}$0.01$} & $43$ & $45.4$ & $43$ & $17.0$  & -- & -- & -- & --  \\
& & & {\color{red!50!blue}$0.005$} & $27$ & $18.3$ & $27$ & $4.9$ & -- & -- & -- & -- \\
  \hline

\multirow{7}{*}{\rotatebox{90}{ZeroConf}}  &  \multirow{3}{*}{$10 ^{4}$}  & \multirow{3}{*}{$0.5$} &  $0.9$  & -- & -- & -- & -- &  $0$ & $0$ & $0$ & $0$   \\ 
& & & $0.52$ & -- & -- & -- & -- & $0$ & $0$ & $0$ & $0$ \\ 
& & & {\color{red!50!blue}$0.45$} & $1$ & $26.5$ & $1$ & $26.5$ & $1$ & $23.2$ & $1$ & $23.2$   \\ 
&  \multirow{4}{*}{$10^9$} & \multirow{4}{*}{${\sim}0.55$} & $0.9$ & -- & -- & -- & -- &  $0$ & $0$ & -- & --   \\ 
&&&$0.75$ & -- & -- & -- & -- &  $0$ & $0$ & -- & --    \\
 & & & $0.52$ & -- & -- & -- & -- & -- & -- & -- & --  \\
& & & {\color{red!50!blue}$0.45$} & $1$ & $27.3$ & $1$ & $24.9$ & $1$ &$24.5$ & $1$ & $26.3$    \\\hline
 \multirow{9}{*}{\rotatebox{90}{Chain}} &  \multirow{4}{*}{$ 10 ^3$}  & \multirow{4}{*}{$0.394$} &  $0.9$ & $0$ & $0$ & $0$ & $0$ & $0$ & $0$ & $0$ & $0$  \\ 
& & & $0$ & $0$ & $0$ & $0$ & $0$ & $0$ & $0$ & --  & --    \\
& & & {\color{red!50!blue}$0.35$} & $1$ & $82.2$ & $1$ & $61.0$ & -- & -- & -- & --  \\
& & & {\color{red!50!blue}$0.3$} & $1$ & $80.0$ & $1$ & $54.7$ & -- & -- & -- & --  \\ 
&  \multirow{3}{*}{$10 ^4$} & \multirow{3}{*}{$0.394$} & $0.9$ & -- & -- & -- & -- &  $0$ & $0$ & $0$ & $0$ \\
& & & $0.4$ & -- & -- & -- & -- & $0$ & $0$ & -- & --    \\
& & & {\color{red!50!blue}$0.3$} & -- & -- & -- & -- & -- & -- & -- & --    \\
&  \multirow{2}{*}{$10 ^{12}$}  & \multirow{2}{*}{$0.394$} &  $0.9$ & -- & -- & -- & -- &  $0$ & $0$ & -- & --   \\
& & & $0.4$ & -- & -- & -- & -- & $0$ & $0$ & -- & --   \\
\hline

\multirow{8}{*}{\rotatebox{90}{Double Chain}} &  \multirow{4}{*}{$ 10 ^3$}  & \multirow{4}{*}{$0.215$} &  $0.9$ & $0$ & $0$ & $0$ & $0$ & $0$ & $0$ & $0$ & $0$  \\  
&&&$0.3$ & $0$ & $0$ & -- & -- & $0$ & $0$  & $0$ & $0$ \\
& & & $0.216$ &  $0$ & $0$ & -- & --  & $0$ & $0$ & -- & --   \\   
& & & {\color{red!50!blue}$0.15$} & -- & -- & -- & -- & -- & -- & -- & --   \\    
%
%
& \multirow{2}{*}{$10 ^{4}$}  & \multirow{2}{*}{$0.22$}& $0.3$ & -- & --  & -- & -- & $0$ & $0$ & $0$ & $0$ \\
& & & {$0.24$} & -- & -- & -- & --  & $0$ & $0$ & -- & --   \\ 
&  \multirow{2}{*}{$10^7$} & \multirow{2}{*}{$2.6\eexp{-4}$} & $4\eexp{-3}$ &  -- & -- & -- & -- & -- & --  & -- & --   \\
& & & $2.7\eexp{-4}$  &  -- & -- & -- & -- & $0$ & $0$ & -- & --  \\ 

\end{tabular}
}

	\label{tab:primaryresults_1_appendix}
\end{table}

\section{Omitted Proofs}

\subsection{Proof of Theorem~\ref{thm:strengthencorrecticcorrect}}
\begin{proof}
	We show that 
	\[
	    I \eeq \icinv{F_0,\ldots,F_{k-1}} \wedge F_{k-1} \leq F_k \wedge \bm{F_{k-1}} \leq F_k
	\]
	is an invariant of the loop in \mainloop. \\ \\
	\noindent
	\emph{Initialization.} At the beginning of the loop, we have $k=1$, $F_0 \eeq \bm{\zeroexp}$, and $F_1 = \oneexp$. $\icinv{F_0}$ holds trivially. 
	Furthermore, we have $F \leq \oneexp$ for every frame $F \in [0, 1]^S$, which gives us
	$F_0 \leq F_1$ and $\bm{F_0} \leq F_1$. \\ \\
	\noindent
	\emph{Consecution.} Assume the loop invariant $I$ holds at the beginning of some iteration. \mainloop~immediately invokes \strengthen~in line $4$. If \strengthen~returns $\success = \false$, we have nothing to show. Otherwise, i.e., if \strengthen~returns $\success=\true$, the specification of \strengthen~(Def.\ \ref{def:spec_strengthen}) implies that $\icinv{F_0,\ldots,F_k}$ holds. Now observe that after line $5$, we have $F_k \leq F_{k+1}$ and $\bm{F_k} \leq F_{k+1}$ since $F_{k+1} = \oneexp$. Since $\propagate$ does not affect the~\pricthree invariants, these facts continue to hold after line $7$. 
	
	If there is an $1 \leq i \leq k$ with $F_i = F_{i+1}$, Lemma~\ref{lem:ic3_inv_eqframes} applies and we have $\prevtmax{\mdp}{\sinit}{\bad} \lleq \thresh$, which is what we have to show.
	
	To conclude our proof that $I$ is a loop invariant, observe that after the assignment $k \leftarrow k+1$ in line $9$, we have
	\[
	     \icinv{F_0,\ldots,F_{k-1}} \wedge F_{k-1} \leq F_k \wedge \bm{F_{k-1}} \leq F_k~,
	\]
	so $I$ continues to hold at the end of that iteration.
\end{proof}
\subsection{Proof of Lemma~\ref{lem:ic3_inv_preserved_by_update}}
\begin{proof}
	Let
	\[
	F_0' \eeq F_0 \subst{s}{\vphantom{\bigl(}\minval{F_0[s],\, \delta}},
	\,\ldots,\, F_k' \eeq F_k \subst{s}{\vphantom{\bigl(}\minval{F_k[s],\, \delta}}~.
	\]
	We consider each invariant separately. \\ \\
	\emph{Initiality.} Since $F_0' = F_0$, initiality is preserved. \\ \\
	\emph{Chain-Property.} Assume $F_0',\ldots,F_k'$ do \emph{not} satisfy the chain-property. Then there is some $i \in \{ 0,\ldots,k-1 \}$ and some state $s'$
	with $F_i'[s'] > F_{i+1}'[s']$. Since $F_0,\ldots,F_k$ satisfy the chain-property and since $s$ is the only state where $F_i'$ and $F_{i+1}'$ might differ, we have $s' = s$.
	There are now two cases: Either $F_{i+1}[s] \leq \delta$ or $F_{i+1}[s] > \delta$. 
	
	If $F_{i+1}[s] \leq \delta$, we also have $F_i[s] \leq \delta$ (chain-property). Hence, we have
	$F_i'[s] = F_{i}[s]$ and $F_{i+1}'[s] = F_{i+1}[s]$. This implies $F_i'[s] \leq F_{i+1}'[s]$ since $F_0,\ldots,F_k$ satisfy the chain-property.
	
	If $F_{i+1}[s] > \delta$, then $F_{i+1}' = \delta$. Since $F_i'[s]$ is the minimum of $F_i[s]$ and $\delta$, we again derive that the chain property holds. \\ \\
	\noindent
	\emph{Frame-Safety.} This case is immediate since we choose the minimum of $F_i[s]$ and $\delta$. This means that if $s = \sinit$ and $\delta > \thresh$, we have
	$F_i'[\sinit] = F_i[\sinit]$ for all $i \in \{0, \ldots, k\}$. Hence, frame-safety continues to hold for $F_0',\ldots,F_k'$. Conversely, if $\delta \leq \thresh$, frame-safety continues to hold, as well. \\ \\
	\noindent
	\emph{Relative Inductiveness.} First notice that the chain-property
	and monotonicity of $\bmsymbol$ imply
	\begin{align*}
	    &\bm{F_{i-1}} \eeleq \bm{F_i} \qquad \textnormal{for all}~ i \in \{1,\ldots,k\}~. \
	\end{align*}
	 This and our assumption $\bm{F_{k-1}}[s] \leq \delta$ yields 
	 \begin{align*}
	    &\bm{F_{i}}[s] \eeleq \delta \qquad \textnormal{for all}~ i \in \{0,\ldots,k-1\}~. \tag{$\varheartsuit$} \label{eq:proof:rel_ind_2}
	 \end{align*}
	Now let $i \in \{1,\ldots,k\}$ and let $s' \in S$. We show $\bm{F_{i-1}'}[s'] \leq F_i'[s']$ by distinguishing the cases $s' = s$ and $s' \neq s$. \\ \\
	\emph{The case $s'=s$.} First observe that
	\begin{align*}
	     &F_{i-1}' \eeleq F_{i-1} \\ 
	     \textnormal{implies} \qquad & \bm{F_{i-1}'} \eeleq \bm{F_{i-1}}   \tag{monotonicity of $\bmsymbol$} \\
	     \textnormal{implies} \qquad & \bm{F_{i-1}'}[s] \eeleq \delta    \tag{\ref{eq:proof:rel_ind_2} and transitivity of $\leq$}
	\end{align*}
	There are now two cases. Either $F_i'[s] = \delta$ or $F_i'[s] = F_i[s]$. If $F_i'[s] = \delta$, we have
	\begin{align*}
	   \bm{F_{i-1}'}[s] \eeleq \delta \eeq F_i'[s] \eeleq F_i'[s]~.
	\end{align*}
	Since $F_i'[s]$ is the minimum of $F_i[s]$ and $\delta$,
	$F_i'[s] = F_i[s]$ implies $F_i[s] \leq \delta$. By the chain property, we then also have $F_{i-1}[s] \leq \delta$. Hence, we have $F_{i-1}' = F_{i-1}$ and $F_i' = F_i$, which gives us
	\begin{align*}
	   & \bm{F_{i-1}'}[s] \\
	   \eeq & \bm{F_{i-1}}[s] \tag{$F_{i-1}' = F_{i-1}$}\\
	   \eeleq & F_i[s] \tag{chain-property for $F_0,\ldots,F_k$} \\
	   \eeq & F_i'[s] \tag{$F_{i}' = F_{i}$} \\
	   \eeleq & F_i'[s] ~.
	\end{align*}
	\emph{The case $s' \neq s$.} In this case, we have $F_i'[s'] = F_i[s']$. Furthermore, monotonicity of $\bmsymbol$ and $F_{i-1}' \leq F_{i-1}$ yields $\bm{F_{i-1}'} \leq \bm{F_{i-1}}$.
	Together, this gives us
	\begin{align*}
	   \bm{F_{i-1}'}[s'] \eeleq \bm{F_{i-1}}[s'] \eeleq F_i[s'] \eeq F_i'[s'] \eeleq F_i'[s']~.
	\end{align*} 
	This completes the proof.
\end{proof}

\subsection{Proof of Lemma~\ref{lem:finallyperfect}}

\begin{proof}
	We prove soundness and completeness separately. \\ \\
	\noindent
	\emph{Soundness.} If \pricthree~terminates inside the loop, then soundness is guaranteed by the soundness of \mainloop~and $\checkrefutation$. Otherwise, i.e., if \pricthree~terminates in line $10$, then $\touched = S$ and our assumption yields
	\begin{align*}
	   \oracle(\sinit) \leq \lambda \quad \text{iff} \quad \prevtmax{\mdp}{\sinit}{\bad} \leq \lambda~.
	\end{align*}
	\emph{Completeness.} Completeness is ensured by the fact that \mainloop~and $\checkrefutation$ terminate. Since $\enlarge$ ensures that $|\touched|$ increases strictly on every iteration, finiteness of the input MDP $\mdp$ implies that the loop in \pricthree~eventually terminates.
	
\end{proof}

\subsection{Proof of Recovery Statement~\ref{recovery:1}}
\begin{proof}
This is a consequence of the well-known fact that~\cite[Chap.\ 10]{BK08}
\[
    \prevtmax{\mc}{\sinit}{\bad} \eeq 0
    \qquad \text{iff} \qquad
    \sinit \not\models \exists \lozenge \bad~.
\]
That is,
the maximal probability to reach a state in $\bad$ from the initial state $\sinit$ is $0$ iff no state in $\bad$ is reachable from $\sinit$ in the graph structure underlying the MDP $\mdp$.
\end{proof}
\subsection{Proof of Recovery Statement~\ref{recovery:2}}
\begin{proof}
	We prove the recovery statement under the reasonable assumption that the heuristic \heuristic is adequate since otherwhise we do not recover standard IC3.
	By Recovery Statement~\ref{recovery:3},
	all obligations $(i,s, \delta)$ spawned by \strengthen~
	satisfy $\delta = 0$. Hence, all frames $F$ computed by \mainloop~satisfy
	\[
	   \forall s \in S \colon F[s] \in \{0,1\}~.
	\]
	For our convenience, we therefore regard frames as \emph{sets of states}.
	
	We now proceed by contradiction. Assume that \mainloop terminates in \textnormal{line~\ref{alg:pric3:zeno}}
	after the $k$th call to $\strengthen$. 
	Our goal is to show that $\oldSubsystem = \subsystem$ implies $F_1 = F_2$, which means that \mainloop has terminated in line~\ref{alg:pric3:checkInductive}, contradicting our assumption. We rely on two observations: \\ \\
	\noindent
	\emph{Observation 1.} Since \mainloop~terminates in \textnormal{line~\ref{alg:pric3:zeno}}, procedure \strengthen~resolved all obligations spawned, i.e., \strengthen~returned $\success = \true$ in every iteration. This implies that every state touched by~\strengthen~is removed from $F_1$:
	Denote by $F_{j,i}$ the frame $F_j$ after the $i$th call to \strengthen~and by $\subsystem_i$ the subsystem returned by the $i$th call to strengthen. We have
	\[
	    F_{1,i} \eeq S \setminus \subsystem_i~.
	\]
	That is, every state touched by~\strengthen~is removed from $F_1$, which is an immediate consequence of the chain-property (cf.\ Definition~\ref{def:invariants}) and the fact that all obligations are resolved. Finally, our assumption that $\oldSubsystem = \subsystem$ holds on iteration $k$ gives us
	\[
	     F_{1,k-1} \eeq F_{1,k}~.
	\]
	\noindent
	\emph{Observation 2.} Denote by $\reach{i}$ the states reachable in at most $i$ steps from initial state $\sinit$.
	As it is the case for standard (reverse) IC3, \strengthen~refines the frames in such a way that
	\[
	    \forall 1 \leq  i \leq k \colon  \forall 1 \leq  j \leq i \colon
	    F_{j,i} \eeq S \setminus \reach{i-j}~.
	\]
	Let us now consider the frames $F_{1,k}$ and $F_{2,k}$.
	Since $F_{1,k-1} = F_{1,k}$ by Observation $1$, we have 
	$\reach{k-2} = \reach{k-1}$ by Observation $2$. This gives us 
	\begin{align*}
	   &F_{1,k} \\
	   \eeq & \reach{k-1} \\
	   \eeq & \reach{k-2} 
	   \tag{by Observation $1$ and $2$} \\
	   \eeq & F_{2,k}~,
	   \tag{by Observation $2$}
	\end{align*}
	which is what we had to show.

\end{proof}
\subsection{Proof of Recovery Statement \ref{recovery:3}}

\begin{proof}
	Let $\thresh = 0$, and let $\getprobs$ be an adequate heuristic. We show that
	\[
	   I \eeq \forall (i,s,\delta)  \in Q \colon \delta =0
	\]
	is an invariant of the loop in Algorithm~\ref{alg:strengthensimple}, which implies the claim. \\ \\
	\noindent
	\emph{Initialization.} Since $\thresh =0$, invariant $I$ holds at the beginning of the loop. \\ \\
	\noindent
	\emph{Consecution.} Assume that $I$ holds at the beginning of some iteration. 
	This implies that after line $3$, we have $\delta = 0$.
	Now assume we spawn new obligations, i.e.\ that line $7$ is executed.
	The check in line $4$ thus evaluted to $\false$, which means $s \not\in \bad$.
	Since $\getprobs$ is adequate, it spawns $\delta_1, \ldots, \delta_n$ such that
	\begin{align*}
	&\bma{a}{\vphantom{\big(}F\subst{s_1}{\delta_1}\ldots\subst{s_n}{\delta_n}}[s] \lleq \delta 
	\tag{since $\getprobs$ is adequate} \\
	\text{iff} \quad &
	  \sum\limits_{i=1}^n \transmatrix(s, a, s_i) \cdot \delta_i \lleq 0
	\tag{ $s \not\in \bad$ and $\succs(s, a) = \{s_1,\ldots,s_n\}$} \\
	\text{iff} \quad & \delta_1 =0 \wedge \ldots \wedge \delta_n =0~.
	\tag{since $\delta_1,\ldots,\delta_n \geq 0$}
	\end{align*}
	This completes the proof.
\end{proof}
\subsection{Proof of Recovery Statement \ref{recovery:4}}
\begin{proof}
	We first show that for every obligation $(i,s,\delta) \in Q$, $\qtouched$
	contains all states on some path from the initial state $\sinit$ to state $s$:
	At the beginning of the loop, this holds trivially since $Q = \{(k, \sinit, \delta)\}$.
	Furthermore, every loop iteration preserves the desired property since:
	\begin{enumerate}
		 \item New obligations consist of successors of states in the queue only (line $8$), and
		 \item state $s$ is reenqueued if obligations for its successors are added to the queue (line $9$).
	\end{enumerate}
	Now let $\thresh =0$, and let $\getprobs_0$ be the adequate heuristic that maps every state to $0$.
	Recall from Recovery Statement \ref{recovery:3} that for every obligation $(i,s,\delta)\in Q$ it holds
	that $\delta =0$. Hence, if the check in line $4$ evalutes to $\true$, we must have $s \in \bad$. Since $\qtouched$ contains a path from the initial state to state $s$, the recovery statement follows.
\end{proof}~
\subsection{Proof of Recovery Statement~\ref{recovery:5}}
\begin{proof}
	This is a consequence of Recovery Statements~\ref{recovery:2} and~\ref{recovery:4}.
	We show that $\pricthree$ terminates after the first call to $\mainloop$ by
	distinguishing the cases where \mainloop~returns $\safeReturn = \true$ and where \mainloop~returns $\safeReturn=\false$. \\ \\
	\noindent
	\emph{ \mainloop~returns $\safeReturn = \true$.}
	In this case, we have nothing to show since $\pricthree$ immediately terminates  in line~$5$. \\ \\
	\noindent
	\emph{\mainloop~returns $\safeReturn = \false$.}
	By Recovery Statement~\ref{recovery:2}, \mainloop~never terminates in line
	\textnormal{line ~\ref{alg:pric3:zeno}}. Hence, \mainloop~has terminated in line~\ref{alg:pric3:unknown}. Now, by Recovery Statement~\ref{recovery:4}, $\subsystem$ contains a path from the initial state $\sinit$ to some state $s \in \bad$. Hence, the probability to reach some bad state in the subMDP $\smdp{\subsystem}$ is larger than $\thresh = 0$, which yields $\checkrefutation(\subsystem)$ to return $\true$ and $\pricthree$ to terminate in line \ref{alg:outer:refute}.
\end{proof}

\end{document}